\RequirePackage{fix-cm}
\documentclass{svjour3}                     

\makeatletter
\let\cl@chapter\undefined
\def\@thmcountersep{.}
\makeatother
\smartqed  
\usepackage{amsmath}
\usepackage{amsfonts}
\usepackage{amssymb}
\usepackage{algorithm}
\usepackage[noend]{algpseudocode}
\usepackage{algorithmicx}
\usepackage{booktabs}
\usepackage{hyperref}
\usepackage{cleveref}
\usepackage{color}
\usepackage{diagbox}
\usepackage[inline]{enumitem}
\usepackage{etoolbox}
\usepackage{float}
\usepackage{mathtools}
\usepackage{multirow}
\usepackage[numbers,sort&compress]{natbib}
\usepackage{tabularx}
\usepackage{thmtools}
\usepackage{xcolor}
\usepackage{comment}

\algrenewcommand\algorithmicindent{0.75em}

\makeatletter
\newcommand{\algmargin}{\the\ALG@thistlm}
\makeatother
\newlength{\whilewidth}
\algdef{SE}[parWHILE]{parWhile}{EndparWhile}[1]
{\parbox[t]{\dimexpr\linewidth-\algmargin}{%
		\hangindent\whilewidth\strut\algorithmicwhile\ #1\ \algorithmicdo\strut}}{\algorithmicend\ \algorithmicwhile}%
\algnewcommand{\parState}[1]{\State%
	\parbox[t]{\dimexpr\linewidth-\algmargin}{\strut #1\strut}}

\crefname{clm}{claim}{claims}
\crefname{figure}{ Fig.}{Fig.}

\newcommand{\bundlep}[3]{({#3},{#1} \backslash {#2})} 
\newcommand{\con}[2]{\nobreak \text{{\em consumers}}^{#2}(#1)\allowbreak}
\newcommand{\ex}[1]{\nobreak B^{#1}\allowbreak}
\newcommand{\gcycle}{Imp}
\newcommand{\impord}[1]{\vartriangleright_{#1}} 
\newcommand{\itemwt}[1]{o_{#1}} 
\newcommand{\hitemwt}[1]{\hat{o}_{#1}}
\newcommand{\pitemwt}[1]{o'_{#1}}
\newcommand{\Leximin}{\text{Leximin-}\allowbreak\text{optimality}}
\newcommand{\leximin}{\text{leximin-}\allowbreak\text{optimality}}
\newcommand{\leximina}{\text{leximin-}\allowbreak\text{optimal}}
\newcommand{\ma}{\mathcal{A}}
\newcommand{\md}{\mathcal{D}}
\newcommand{\mr}{\mathcal{R}}
\newcommand{\mP}{\mathcal{P}}
\newcommand{\mtap}{\text{MTRA}} 
\newcommand{\mtaps}{\text{MTRAs}} 
\newcommand{\mps}{\text{MPS}}
\newcommand{\ngc}{\text{no-}\allowbreak\text{generalized-}\allowbreak\text{cycle}}

\newcommand{\prefintype}[2]{\succ_{#1}^{#2}} 
\newcommand{\prog}[1]{\nobreak\text{$\rho$}^{#1}\allowbreak}
\newcommand{\ram}{fractional} 
\newcommand{\dec}{decompoability}
\newcommand{\deca}{decomposable}
\newcommand{\s}{\text{{\em supply}}}
\newcommand{\sd}[1]{~\nobreak \succ^{sd}_{#1}\allowbreak~}
\newcommand{\nsd}[1]{~\nobreak \not\succ^{sd}_{#1}\allowbreak~}
\newcommand{\sdef}{\text{sd-}\allowbreak\text{envy-}\allowbreak\text{freeness}}
\newcommand{\sdefa}{\text{sd-}\allowbreak\text{envy-}\allowbreak\text{free}}
\newcommand{\epopt}{\text{ex-post-}\allowbreak\text{efficiency}}
\newcommand{\sdopt}{\text{sd-}\allowbreak\text{efficiency}}
\newcommand{\sdopta}{\text{sd-}\allowbreak\text{efficient}}
\newcommand{\sdsp}{\text{sd-}\allowbreak\text{weak-}\allowbreak\text{strategy}\allowbreak\text{proofness}}
\newcommand{\sdspa}{\text{sd-}\allowbreak\text{weak-}\allowbreak\text{strategy}\allowbreak\text{proof}}
\newcommand{\sdssp}{\text{sd-}\allowbreak\text{strategy}\allowbreak\text{proofness}}

\newcommand{\supply}{\text{supply}}
\newcommand{\topb}[2]{\nobreak top^{#2}(#1)\allowbreak}
\newcommand{\type}[2]{D_{#1}({#2})} 
\newcommand{\ucs}{U}			

\newcommand{\vx}{{\bf{x}}}

\newcommand{\sdwef}{\text{sd-}\allowbreak\text{weak-}\allowbreak\text{envy-}\allowbreak\text{freeness}}
\newcommand{\sdwopt}{\text{sd-}\allowbreak\text{weak-}\allowbreak\text{efficiency}}

\newcommand{\lps}{LexiPS}
\newcommand{\vu}{{\bf{u}}}
\newcommand{\vv}{{\bf{v}}}

\newcommand{\tabincell}[2]{\begin{tabular}{@{}#1@{}}#2\end{tabular}}
\newcommand{\srd}[2]{{#1}^{#2}} 

\newcommand{\ld}[1]{\succ_{#1}^{lexi}}
\newcommand{\Iof}{\text{Item-wise }\allowbreak\text{ordinal }\allowbreak\text{fairness}}
\newcommand{\iof}{\text{item-wise }\allowbreak\text{ordinal }\allowbreak\text{fairness}}
\newcommand{\iofa}{\text{item-wise }\allowbreak\text{ordinal }\allowbreak\text{fair}}
\newcommand{\lexopt}{\text{lexi-}\allowbreak\text{efficiency}}
\newcommand{\lexopta}{\text{lexi-}\allowbreak\text{efficient}}

\spnewtheorem{clm}{Claim}{\bfseries}{\normalfont}
\spnewtheorem{clmprf}{Proof}{\bfseries}{\normalfont}
\usepackage{graphicx}
\usepackage{mathptmx}      
\usepackage[misc]{ifsym}
\journalname{}
\begin{document}

\title{Probabilistic Serial Mechanism for Multi-Type Resource Allocation
}


\author{Xiaoxi Guo \textsuperscript{1}
        \and
       Sujoy Sikdar \textsuperscript{2}
        \and
        Haibin Wang \textsuperscript{1}
        \and
        Lirong Xia \textsuperscript{3}
        \and
        Yongzhi Cao \textsuperscript{1,*}
        \and
        Hanpin Wang \textsuperscript{1}
}

\authorrunning{Xiaoxi Guo 
        \and
       Sujoy Sikdar 
        \and
        Haibin Wang 
        \and
        Lirong Xia 
        \and
        Yongzhi Cao 
        \and
        Hanpin Wang 
        } 
            %

\institute{
Xiaoxi Guo \at
\email{guoxiaoxi@pku.edu.cn}  \\
\and
Sujoy Sikdar\at
\email{sujoy@wustl.edu}\\
\and
Haibin Wang\at
\email{beach@pku.edu.cn}\\
\and
Lirong Xia\at
\email{xialirong@gmail.com}\\
\and
\Letter Yongzhi Cao\at
\email{caoyz@pku.edu.cn}\\
\and
Hanpin Wang\at
\email{whpxhy@pku.edu.cn}\\
\at
 {1} Key Laboratory of High Confidence Software Technologies (MOE), Department of Computer Science and Technology, Peking University, China.
 \at
 {2} Computer Science \& Engineering, Washington University in St. Louis.
  \at
 {3} Department of Computer Science, Rensselaer Polytechnic Institute.
}

\date{Received: date / Accepted: date}

\maketitle

\begin{abstract}

In  {\em multi-type resource allocation (\mtap{}) problems}, there are $p\ge 2$ {\em types} of items, and $n$ agents, who each demand one unit of items of each type, and have {\em strict linear preferences} over {\em bundles} consisting of one item of each type.

For \mtaps{} with indivisible items, our first result is an impossibility theorem that is in direct contrast to the single type ($p=1$) setting: No mechanism, the output of which is always {\em \deca{}} into a probability distribution over discrete assignments (where no item is split between agents), can satisfy both \sdopt{} and \sdef{}.
To circumvent this impossibility result, we consider the natural assumption of lexicographic preference, and provide an extension of the {\em probabilistic serial (PS)}, called {\em lexicographic probabilistic serial (\lps{})}. We prove that \lps{} satisfies \sdopt{} and \sdef{}, retaining the desirable properties of PS. Moreover, \lps{} satisfies \sdsp{} when agents are not allowed to misreport their importance orders.

For \mtaps{} with divisible items, we show that the existing {\em multi-type probabilistic serial (\mps{})} mechanism satisfies the stronger efficiency notion of \lexopt{}, and is \sdefa{} under strict linear preferences, and \sdspa{} under lexicographic preferences. We also prove that \mps{} can be characterized both by \leximin{} and by \iof{}, and the family of eating algorithms which \mps{} belongs to can be characterized by \ngc{} condition.

\end{abstract}

\section{Introduction}
In this paper, we focus on extensions of the celebrated {\em probabilistic serial} (PS) mechanism~\cite{Bogomolnaia01:New} for the classical resource allocation problem~\cite{Abdulkadiroglu99:House,Bogomolnaia01:New,Chevaleyre06:Issues,Moulin95:Cooperative}, to the {\em multi-type resource allocation problem} (\mtap{})~\cite{Mackin2016:Allocating}. An~\mtap{} involves $n$ agents, $p\ge 2$ {\em types} of items which are not interchangeable, one unit each of $n$ items of each type. Each agent demands a {\em bundle} consisting of one item of each type, and has {\em strict preferences} over all bundles. \mtaps{} may involve {\em divisible} items, like land and water resources~\cite{Segal17:Fair}, or computational resources such as CPU, memory, and storage in cloud computing~\cite{Ghodsi11:Dominant,Ghodsi12:Multi,Grandl15:Multi}. Items may also be {\em indivisible}, where each item must be assigned fully to a single agent, like houses and cars~\cite{Sikdar18:Top,Sikdar2017:Mechanism}, or research papers and time slots in a seminar class~\cite{Mackin2016:Allocating}.

Efficient and fair resource allocation for a single type of items ($p=1$) is well studied~\cite{Bogomolnaia01:New,Hylland79:Efficient,zhou1990conjecture,Abdulkadiroglu98:Random,Shapley74:Cores,Abdulkadiroglu99:House,Chevaleyre06:Issues,Moulin95:Cooperative,moulin2018fair}.
Our work follows the line of research initiated by~\citet{Bogomolnaia01:New}, who proposed the probabilistic serial (PS) mechanism. The PS mechanism outputs a fractional assignment in multiple rounds, by having all agents simultaneously “eat” shares of their favorite remaining items at a uniform and equal rate, until one of the items is exhausted, in each round.

PS is a popular prototype for mechanism designers due to the following reasons:
\begin{enumerate*}[label=(\roman*)]
    \item {\em decomposability}: PS can be applied to allocate both divisible and indivisible items, since fractional assignments are always {\em \deca{}} when $p=1$, as guaranteed by the Birkhoff-von Neumann theorem~\cite{Birkhoff1946:Three,Neumann1953:A-certain}. In other words, a fractional assignment can be represented as a probability distribution over ``discrete'' assignments, where no item is split among agents.
    \item {\em efficiency and fairness}: PS satisfies \sdopt{} and \sdef{}, which are desirable efficiency and fairness properties respectively based on the notion of {\em stochastic dominance}: Given a strict preference over the items, an allocation $p$ {\em stochastically dominates} $q$, if at every item $o$, the total share of $o$ and items strictly preferred to $o$ in $p$, is at least the total share of the same items in $q$.
\end{enumerate*}

Unfortunately, designing efficient and fair mechanisms for \mtap{} with $p\ge 2$ types is more challenging, especially because direct applications of PS to \mtap{} fail to satisfy the two desirable properties discussed above.

First, \dec{} (property (\romannumeral1)) above relies on the \dec{} of fractional assignments, which not always holds for \mtap{} as in the following simple example.
\begin{example}
\label{eg:undecomposable}
Consider the \mtap{} with two agents, $1$ and $2$, two types of items, food ($F$) and beverages ($B$), and two items of each type $\{1_F,2_F\}$, and $\{1_B, 2_B\}$ respectively.
We demonstrate how the fractional assignment $P$ below, where agent $1$ gets $0.5$ share of $1_F1_B$ and $0.5$ share of $2_F2_B$ is not \deca{}.

\vspace{1em}\noindent\begin{minipage}{\linewidth}
	\begin{minipage}{0.4\linewidth}
     \begin{center}
        \centering
        \begin{tabular}{|c|cccc|}
            \hline\multirow{2}{*}{Agent} & \multicolumn{4}{c|}{$P$}\\\cline{2-5}
             & $1_F1_B$ & $1_F2_B$ & $2_F1_B$ & $2_F2_B$ \\\hline
            1 & 0.5 & 0 & 0 & 0.5 \\
            2 & 0 & 0.5 & 0.5 & 0 \\\hline
        \end{tabular}
    \end{center}
	\end{minipage}
	\hspace{0.1\linewidth}
	\begin{minipage}{0.4\linewidth}
	\begin{center}
        \centering
        \begin{tabular}{|c|cccc|}
            \hline\multirow{2}{*}{Agent} & \multicolumn{4}{c|}{$P'$}\\\cline{2-5}
             & $1_F1_B$ & $1_F2_B$ & $2_F1_B$ & $2_F2_B$ \\\hline
            1 & 1 & 0 & 0 & 0 \\
            2 & 0 & 0 & 0 & 1 \\\hline
        \end{tabular}
    \end{center}
	\end{minipage}
	\end{minipage}\vspace{1em}

Obviously, the assignment $P'$ as above is the only assignment where $1_F1_B$ is allocated fully to agent $1$.
Since agent $1$ acquires $1_F1_B$ with $0.5$ shares in $P$, the probability for $P'$ ought to be 0.5, and therefore $2_F2_B$ should be allocated to agent $2$ with $0.5$ shares in $P$ accordingly.
However, agent $2$ is not allocated $2_F2_B$ in $P$ actually.
Thus $P$ is not \deca{}.
\end{example}

A natural idea is to decompose MTRA into $p$ single-type instances, one for each type of items, and then apply PS or other mechanisms separately to each of them. Unfortunately, this does not work because first it is unclear how to decompose agents' combinatorial preferences over bundles into separable preferences over items of the same type, and more importantly, even when there is a natural way to do so, e.g.~when agents' preferences are {\em separable}, and {\em lexicographic}, meaning that the agent has an importance order over types to compare bundles. The following example shows that the fairness and efficiency properties (\romannumeral2) above do not hold anymore.
\begin{example}\label{eg:ps}
We continue to use the \mtap{} above and assume that agents' preferences over $\{1_F,2_F\}\times\{1_B,2_B\}$ are as below.
    \begin{center}
        \begin{tabular}{|c|c|} \hline
            Agent & Preferences \\\hline
            1 & $1_F1_B\succ_11_F2_B\succ_12_F1_B\succ_12_F2_B$ \\
            2 & $1_F1_B\succ_22_F1_B\succ_21_F2_B\succ_22_F2_B$ \\ \hline
        \end{tabular}
    \end{center}
We note that both agents prefer $1_F$ to $2_F$, and $1_B$ to $2_B$ (separable preferences). Agent $1$ considers $F$ to be more important than $B$, while agent $2$ considers $B$ to be more important.
In this way we can decompose this \mtap{} into two single type resource allocation problems for $F$ and $B$ respectively.
It is easy to see that for each single type the only \sdopta{} and \sdefa{} assignment is to give both agents $0.5$ shares of each item, yielding the \deca{} fractional assignment $Q$ by the mutually independence of each type.
However, $Q$ is inefficient, as $Q'$ stochastically dominates $Q$ from both agents' perspectives:

\vspace{1em}\noindent\begin{minipage}{\linewidth}
	\begin{minipage}{0.4\linewidth}
    \begin{center}
        \centering
        \begin{tabular}{|c|cccc|}
            \hline\multirow{2}{*}{Agent} & \multicolumn{4}{c|}{$Q$}\\\cline{2-5}
             & $1_F1_B$ & $1_F2_B$ & $2_F1_B$ & $2_F2_B$ \\\hline
            1 & 0.25 & 0.25 & 0.25 & 0.25 \\
            2 & 0.25 & 0.25 & 0.25 & 0.25 \\\hline
        \end{tabular}
    \end{center}
	\end{minipage}
	\hspace{0.1\linewidth}
	\begin{minipage}{0.4\linewidth}
    \begin{center}
        \centering
        \begin{tabular}{|c|cccc|}
            \hline\multirow{2}{*}{Agent} & \multicolumn{4}{c|}{$Q'$}\\\cline{2-5}
             & $1_F1_B$ & $1_F2_B$ & $2_F1_B$ & $2_F2_B$ \\\hline
            1 & 0.25 & 0.5 & 0 & 0.25 \\
            2 & 0.25 & 0 & 0.5 & 0.25 \\\hline
        \end{tabular}
    \end{center}
	\end{minipage}
	\end{minipage}\vspace{1em}
	
\end{example}

As we have observed, the two desirable properties of PS for single type resource allocation no longer obviously hold for \mtaps{}.
Recently,~\citet{Wang19:Multi} proposed {\em  multi-type probabilistic serial (\mps{})} mechanism as an extension of PS for \mtaps{} with {\em divisible} items, and proved that \mps{} is \sdopta{} for general partial preference, \sdefa{} for CP-net preferences~\cite{Boutilier04:CP}, and \sdspa{} for CP-net preferences with the shared dependency graph. However, it is unclear whether \mps{} can be applied to the allocation of indivisible items because the outcome may not be \deca{}. This leaves the following natural open question:

{\em How to design efficient and fair mechanisms for \mtaps{} with indivisible and divisible items?}
\footnote{Note that for indivisible items, the (fractional) output of a mechanism must be decomposable.}

\vspace{0.5em}
\paragraph{\bf Our Contributions.} 
For \mtaps{} with indivisible items, unfortunately, our impossibility theorem ({\bf\Cref{thm:imp}}) shows that no mechanism which satisfies \sdopt{} and \sdef{} is guaranteed to always output \deca{} assignments, if agents' preferences are allowed to be any linear orders over bundles.
Fortunately, when agents' preferences are {\em lexicographic}, the impossibility theorem can be circumvented. To this end, we propose {\em lexicographic probabilistic serial mechanism (\lps)} and prove that it satisfies many desirable properties of PS: it is guaranteed to output a \deca{} assignment, satisfy \sdopt{} and \sdef{} ({\bf\Cref{thm:lps}}), and satisfies \sdsp{}, when agents do not lie about their importance orders over types ({\bf\Cref{thm:lpssp}}).



For \mtaps{} with divisible items, we
show that when agents' preferences are linear orders over all bundles of items, the \mps{} mechanism proposed by~\citet{Wang19:Multi} satisfies \lexopt{} ({\bf\Cref{thm:mpslexopt}}) which is a stronger notion of efficiency than \sdopt{}. Indeed, we show that \lexopt{} implies the {\em \ngc{}} condition, which is a sufficient condition for \sdopt{} (similarly to~\citet{Wang19:Multi}), but not a necessary one ({\bf\Cref{prop:gc}}).
We also prove that {\em every} assignment satisfying \ngc{} can be computed by some algorithm in the family of {\em eating} algorithms ({\bf\Cref{thm:eps}}), of which \mps{} is a member.
Importantly, \mps{} retains \sdef{} ({\bf\Cref{prop:mps}}), and when agents' preferences are further assumed to be lexicographic, \mps{} satisfies \sdsp{} ({\bf\Cref{thm:mpssp}}).
Finally, we characterize \mps{} by \leximin{} and \iof{}, respectively ({\bf\Cref{thm:char}}).
However, the output of \mps{} is not always \deca{} (\Cref{rm:mpsnreal}) even under lexicographic preference, making it unsuitable for \mtaps{} with indivisible items.
%



\vspace{0.5em}
\paragraph{\bf Related Work and Discussions.}
Our paper provides the first results on designing efficient and fair mechanisms for \mtaps{} {\em with indivisible items}, to the best of our knowledge. Despite our impossibility theorem ({\bf\Cref{thm:imp}}), our \lps{} mechanism and its properties deliver the following positive message: {\em it is possible to design efficient and fair mechanism for indivisible items under the natural domain restriction of lexicographic preferences}.

Our results on properties of \mps{} are complementary, and not directly comparable, to the results by~\citet{Wang19:Multi} because we assume  that agents' preferences are linear orders over bundles of items, whereas~\citet{Wang19:Multi} assumed partial orders. In particular, we prove that \mps{} satisfies \lexopt{} which is a stronger notion that \sdopt{} for the unrestricted domain of linear orders, and we prove that \mps{} satisfies \sdsp{} when agents' preferences are lexicographic (w.r.t.~possibly different importance orders). In contrast,~\citet{Wang19:Multi} prove that \mps{} satisfies \sdopt{} for the unrestricted domain of partial orders, and satisfies \sdsp{} when agents' preferences are CP-nets with a common dependency structure.


\mtaps{} was first introduced and discussed by~\citet{Moulin95:Cooperative}, and was recently explicitly formulated in the form presented in this paper by~\citet{Mackin2016:Allocating}, who
provided a characterization of serial dictatorships satisfying strategyproofness, neutrality, and non-bossiness for \mtaps{}. In a similar vein, \citet{Sikdar18:Top,Sikdar2017:Mechanism}~considered multi-type housing markets~\cite{Moulin95:Cooperative}. A related problem setting is one where agents may demand multiple units of items.~\citet{Hatfield09:Strategy-proof} considered that agents have multi-unit demands, but their setting has a different combinatorial structure to ours'. \citet{Fujita2015:A-Complexity}~considered the exchange economy with multi-unit consumption. However, in each of these works, items are never shared, and agents must be assigned a whole bundle. The work by~\citet{Ghodsi11:Dominant} is an exception, as they consider the problem of allocating multiple types of divisible resources by shares, but is not comparable to ours because resources of the same type are indistinguishable in their setting. \mtaps{} with divisible items may also be viewed as a version of the cake-cutting problem~\cite{Steinhaus48:Problem,brams1996fair,robertson1998cake,even1984note,edmonds2006cake,procaccia2015cake}, with $p$ cakes, and agents having ordinal preferences over combinations of pieces from each cake.

The lexicographic preference is a natural restriction on preference domain in resource allocation~\cite{Sikdar18:Top,Sikdar2017:Mechanism,Fujita2015:A-Complexity}, and combinatorial voting~\citep{Lang12:Aggregating,Xia10:Strategy,Booth10:Learning}.~\citet{Saban14:Note} showed that PS is efficient, envy-free, and strategy-proof under lexicographic preference on allocations.~\citet{Fujita2015:A-Complexity} considered the allocation problem which allows agents to receive multiple items and agents rank the groups of items lexicographically. Our work follows in this research agenda of natural domain restrictions on agents' preferences to circumvent impossibility results in guaranteeing efficiency and fairness.

\mtaps{} belongs to a more general line of research on mechanism design known as multi-agent resource allocation (see \citet{Chevaleyre06:Issues} for a survey), where literature focuses on the problem where items are of a single type.
Early research focused mainly on developing ``discrete'' mechanisms for indivisible items, where each item is assigned fully to a single agent~\cite{Gale62:College,Kojima09:Axioms,Shapley74:Cores,Roth77:Weak,Abdulkadiroglu99:House}. However, discrete mechanisms often fail to simultaneously satisfy efficiency and fairness.

Fractional mechanisms simultaneously provide stronger efficiency and fairness guarantees. For example, the random priority (RP) mechanism~\cite{Abdulkadiroglu98:Random} outputs fractional assignments, and satisfies \epopt{}, \sdwef{}, and \sdssp{}.
Such fractional mechanisms can be applied to both divisible and indivisible items in the single type setting ($p=1$), due to the Birkhoff-Von Neumann theorem which implies that every fractional assignment is \deca{}.

\citet{Bogomolnaia01:New}~proposed the PS mechanism, a fractional mechanism satisfying \sdopt{}, \sdef{}, and \sdsp{}. PS uniquely possesses the important properties of \sdopt{} and \sdef{} with the restriction of {\em bounded invariance}~\cite{Bogomolnaia12:Probabilistic,Bogomolnaia15:Random}, and is the only mechanism satisfying ordinal fairness and non-wastefulness~\cite{Hashimoto14:Two}. Besides, \citet{Bogomolnaia15:Random}~characterized PS with leximin maximizing the vector describing cumulative shares at each item~\cite{Aziz14:Generalization,bogomolnaia2005collective}, which reflects that PS is egalitarian in attempting to equalize agents’ shares of their top ranked choices. The remarkable properties of PS has encouraged several extensions: to the full preference domain, allowing indifferences~\cite{Katta06:Solution,Heo15:Characterization}, to multi-unit demands~\cite{Heo14:Probabilistic}, and to housing markets~\cite{Athanassoglou11:House,Yilmaz2009:Random}. In the settings above, PS usually only retains some of its original properties and loses stategyproofness~\cite{Katta06:Solution,Heo14:Probabilistic,Yilmaz10:Probabilistic,Athanassoglou11:House}.

\vspace{0.5em}
\paragraph{\bf Structure of the paper.} The rest of the paper is organized as follows.
In Section~\ref{Preliminaries}, we define the \mtap{} problem, and provide definitions of desirable efficiency and fairness properties.
Section~\ref{Impossibility Result} is the impossibility result for \mtaps{} with indivisible items.
In Section~\ref{LPS}, we propose \lps{} for \mtaps{} with indivisible items under lexicographic preferences, which satisfies \sdopt{} and \sdef{}, and it is \sdspa{} when agents do not lie about importance orders.
In Section~\ref{MPS}, we show the properties of \mps{} for \mtaps{} with divisible items under linear preferences, and provide two characterizations for \mps{}.
In Section~\ref{Conclusion}, we summarize the contributions of our paper, and discuss directions for future work.

\section{Preliminaries}
\label{Preliminaries}

Let $N=\{1,\dots,n\}$ be the set of agents, and $M=D_1\cup\dots\cup D_p$ be the set of items.
For each $i\le p$, $D_i$ is a set of $n$ items of  type $i$, and for all $\hat{i}\neq i$, we have $D_i\cap D_{\hat{i}}=\emptyset$.
There is one unit of {\em \supply{}} of each item in $M$.
We use $\md=D_1\times\dots\times D_p$ to denote the set of {\em bundles}.
Each bundle $\vx\in\md$ is a $p$-vector and each component refers to an item of each type.
We use $o \in \vx$ to indicate that bundle $\vx$ contains item $o$.
In an \mtap{}, each agent demands, and is allocated one unit of item of each type.

A {\em preference profile} is denoted by $R=(\succ_j)_{j\le n}$, where $\succ_j$ represents agent $j$'s preference as a {\em strict linear preference}, i.e. the strict linear order over $\md$. Let $\mr$ be the set of all possible preference profiles.

A {\em \ram{} allocation} is a $|\md|$-vector, describing the fractional share of each bundle allocated to an agent.
Let $\Pi$ be the set of all the possible \ram{} allocations. For any $p\in\Pi$, $\vx\in\md$, we use $p_{\vx}$ to denote the share of $\vx$ assigned by $p$.
A {\em \ram{} assignment} is a $n\times|\md|$-matrix $P=[p_{j,\vx}]_{j\le n, \vx\in\md}$, where
\begin{enumerate*}[label=(\roman*)]
    \item for each $j\le n, \vx\in\md$, $p_{j,\vx}\in[0,1]$ is the fractional share of $\vx$ allocated to agent $j$,
	\item for every $j\le n$, $\sum_{\vx\in\md}p_{j,\vx}=1$, fulfilling the demand of each agent,
	\item for every $o\in M$, $S_o=\{\vx:\vx\in\md \text{ and }o\in\vx\}$, $\sum_{j\le n,\vx \in S_o}p_{j,\vx}=1$, respecting the unit supply of each item.
\end{enumerate*}
For each $j\le n$, row $j$ of $P$, denoted by $P_j$ represents agent $j$'s \ram{} allocation under $P$.
We use $\mP$ to denote the set of all possible \ram{} assignments.
A {\em discrete assignment} $A$, is an assignment where each agent is assigned a unit share of a bundle, and each item is fully allocated to some agent\footnote{For for indivisible items, discrete assignments refer to deterministic assignments in the papers about randomization.}. It follows that a discrete assignment is represented by a matrix where each element is either $0$ or $1$. We use $\ma$ to denote the set of all discrete assignment matrices.

A {\em mechanism} $f$ is a mapping from preference profiles to \ram{} assignments.
For any profile $R\in\mr$, we use $f(R)$ to refer to the \ram{} assignment output by $f$, and for any agent $j\le n$ and any bundle $\vx\in\md$, $f(R)_{j,\vx}$ refers to the value of the element of the matrix indexed by $j$ and $\vx$.

\subsection{\textbf{Desirable Properties}}
\label{Properties}


\noindent We use the notion of stochastic dominance to compare fractional assignments, and define desired notions of efficiency and fairness in the \mtap{} setting from~\citet{Wang19:Multi}.


\begin{definition}{\bf(stochastic dominance~\cite{Wang19:Multi})}\label{dfn:sd}
	Given a preference $\succ$ over $\md$, the {\em stochastic dominance} relation associated with $\succ$, denoted $\sd{\null}$, is a partial ordering over $\Pi$ such that for any pair of \ram{} allocations $p,q\in\Pi$, $p$ {\em weakly stochastically dominates} $q$, denoted $p\sd{\null} q$, if and only if for every $\vx\in\md$, $\sum_{\hat\vx\in\ucs(\succ,\vx)}p_{\hat\vx}\ge\sum_{\hat\vx\in\ucs(\succ,\vx)}q_{\hat\vx}$, where  $\ucs(\succ,\vx)=\{\hat\vx:\hat\vx \succ \vx\}\cup\{\vx \}$.
\end{definition}


The stochastic dominance order can also be extended to fractional assignments.
For $P,Q\in\mP,j\le n$, we assume that agent $j$ only cares about her own allocations $P_j,Q_j$.
If $P_j\sd{j} Q_j$, it will weakly prefer $P$ to $Q$, i.e. $P\sd{j} Q$.
Therefore we say that $P$ weakly stochastically dominates $Q$, denoted $P\sd{\null}Q$, if $P\sd{j}Q$ for any $j\leq n$.
It is easy to prove that $P\sd{j} Q,Q\sd{j} P$ if and only if $P_j=Q_j$.

\begin{definition}{\bf(\sdopt{}~\cite{Wang19:Multi})}\label{dfn:sdopt}
	Given an \mtap{} $(N,M,R)$, a \ram{} assignment $P$ is \sdopta{} if there is no other \ram{} assignment $Q\neq P$ such that $Q\sd{j}P$ for every $j\le n$.
    Correspondingly, if for every $R\in\mr$, $f(R)$ is \sdopta{}, then mechanism $f$ satisfies \sdopt.
\end{definition}

\begin{definition}{\bf(\sdef{}~\cite{Wang19:Multi})}\label{dfn:sdef}
	Given an \mtap{} $(N,M,R)$, a \ram{} assignment $P$ is \sdefa{} if for every pair of agents $j,\hat j\le n$, $P_j\sd{j}P_{\hat j}$.
    Correspondingly, if for every $R\in\mr$, $f(R)$ is \sdefa{}, then mechanism $f$ satisfies \sdef.
\end{definition}


\begin{definition}{\bf(\sdsp{}~\cite{Wang19:Multi})}\label{dfn:sdsp}
	Given an \mtap{} $(N,M,R)$, a mechanism $f$ satisfies \sdsp{} if for every profile $R\in\mr$, every agent $j\le n$, every $R'\in\mr$ such that $R'=(\succ'_j,\succ_{-j})$, where $\succ_{-j}$ is the preferences of agents in the set~$N\setminus\{j\}$, it holds that:
    \begin{equation*}
      f(R')\sd{j}f(R)\implies f(R')_j=f(R)_j.
    \end{equation*}
\end{definition}

Besides stochastic dominance, we introduce the {\em lexicographic dominance} relation to compare pairs of fractional allocations, by comparing the components of their respective vector representations one by one according to the agent's preference.

\begin{definition}
{\bf(lexicographic dominance)} Given a preference $\succ$, and a pair of allocations $p$ and $q$, $p$ lexicographically dominates $q$, denoted $p\ld{}q$, if and only if there exist a bundle $\vx$ such that $p_{\vx}>q_{\vx}$, and for each $\hat\vx\succ\vx$, $p_{\hat\vx}\ge q_{\hat\vx}$.
\end{definition}

Just like stochastic dominance, given assignments $P$ and $Q$, we say $Q \ld{} P$ if $Q_j \ld{j} P_j$ ($Q \ld{j} P$ for short) for every agent $j$. Note that stochastic dominance implies lexicographic dominance, but lexicographic dominance does not imply stochastic dominance.

\begin{definition}
    {\bf(\lexopt{})} Given a preference profile $R$, the fractional assignment $P$ is \lexopta{} if there is no $Q \in \mP$ s.t. $Q \ld{} P$.
    A fractional assignment algorithm $f$ satisfies \lexopt{} if $f(R)$ is \lexopta{} for every $R \in \mr$.
\end{definition}


\Leximin{} requires that the output of a mechanism leximin maximizes the vector describing cumulative shares at each item~\cite{Aziz14:Generalization,bogomolnaia2005collective}, which reflects the egalitarian nature of the mechenism in attempting to equalize agents’ shares of their top ranked choices.
\begin{definition}{\bf(\leximin{})}\label{dfn:leximin}
For any vector $\vu$ of length $k$, let $\vu^*=(u_1^*,u_2^*,\dots,u_k^*)$, be its transformation into the $k$-vector of $\vu$'s components sorted in ascending order. Let $L$ be the {\em leximin relation}, where for any two vectors $\vu, \vv$, we say $(\vu,\vv)\in L$, if  $\vu^*\ld{}\vv^*$.
For any \ram{} assignment $P$, let $\vu^P=(u^P_{j,\vx})_{j\le n, \vx\in\md}$, where for each agent $j\le n$, and bundle $\vx\in\md$, $u^P_{j,\vx}=\sum_{\hat\vx\in\ucs(\succ_j,\vx)}{p_{j,\hat\vx}}$. A \ram{} assignment $P$ is {\em \leximina{}}, if for every other assignment $Q$, $(\vu^P,\vu^Q)\in L$.
\end{definition}

\Iof{} involves the comparison of the cumulative share.
In contrast to \sdef{}, the upper contour sets in \iof{} depend on the different preferences, and the bundles to determine the sets only need to share a certain item.
\begin{definition}
    {\bf (\iof{})} The fractional assignment $P$ which is \iofa{} satisfies that: For any $\vx$ with positive shares for agent $j$ s.t. $p_{j,\vx}>0$, there exists $o\in\vx$ such that for any $\hat\vx$ containing $o$ and an arbitrary agent $k$, if $p_{k,\hat\vx}>0$, then $\sum_{\vx'\in\ucs(\succ_k,\hat\vx)}{p_{k,\vx'}}\le\sum_{\vx'\in\ucs(\succ_j,\vx)}{p_{j,\vx'}}$.
     A fractional assignment algorithm $f$ satisfies \iof{} if $f(R)$ is \iofa{} for every $R \in \mr$.
\end{definition}

\section{Efficiency and Fairness for \mtaps{} with Indivisible Items}
\label{Impossibility Result}
In this section, we show an impossibility result in \Cref{thm:imp} that no mechanism satisfying \sdef{} and \sdopt{} is guaranteed to output \deca{} assignments.
This is unlike the case of resource allocation problems with a single type of items, where by the Birkhoff-von Neumann theorem, every \ram{} assignment is \deca{}, i.e. every \ram{} assignment can be decomposed like
\begin{equation*}
    P=\sum_{A_k\in\ma}{\alpha_k\times A_k}.
\end{equation*}
Here each $A_k$ is a discrete assignment that assigns every item wholly to some agent. Observe that $\sum{\alpha_k}=1$. It follows that such a \deca{} assignment can be applied to allocating indivisible items, by issuing a lottery for $A_k$, and $\alpha_k$ is its probability being selected among all the discrete assignments.
It is not necessarily so in \mtaps{}, which leads to the impossibility result.


\begin{theorem}\label{thm:imp}
	For any \mtaps{} with $p\ge 2$, where agents are allowed to submit any linear order over bundles, no mechanism that satisfies \sdopt{} and \sdef{} always outputs \deca{} assignments.
\end{theorem}
\begin{proof}
	Suppose for the sake of contradiction that there exists a mechanism $f$ satisfying \sdopt{} and \sdef{} and $f(R)$ is always \deca{} for any $R\in\mr$.
    Let R be the following preference and $Q=f(R)$.
    \begin{center}
        \centering
        \begin{tabular}{|c|c|}
            \hline Agent & Preferences \\\hline
            1 & $1_F1_B\succ_11_F2_B\succ_12_F2_B\succ_12_F1_B$ \\
            2 & $1_F2_B\succ_22_F1_B\succ_21_F1_B\succ_22_F2_B$ \\\hline
        \end{tabular}
    \end{center}

    We show that if $Q$ is \sdefa{} and \deca{}, it fails to satisfy \sdopt{}.
	There are only four discrete assignments which assign $1_F1_B,1_F2_B,2_F1_B,2_F2_B$ to agent $1$ respectively.
	Since $Q$ is \deca{}, it can be represented as the following assignment.
	We also provide an assignment $P$ which is not \deca{} since it does not satisfy the constraints for $Q$.

    \vspace{1em}\noindent
    \begin{minipage}{\linewidth}
	\begin{minipage}{0.4\linewidth}
    \begin{center}
        \centering
        \begin{tabular}{|c|cccc|}
            \hline\multirow{2}{*}{Agent} & \multicolumn{4}{c|}{$P$}\\\cline{2-5}
             & $1_F1_B$ & $1_F2_B$ & $2_F1_B$ & $2_F2_B$ \\\hline
            1 & 0.5 & 0 & 0 & 0.5 \\
            2 & 0 & 0.5 & 0.5 & 0 \\\hline
        \end{tabular}
    \end{center}
	\end{minipage}
	\hspace{0.1\linewidth}
	\begin{minipage}{0.4\linewidth}
    \begin{center}
        \centering
        \begin{tabular}{|c|cccc|} \hline
            \multirow{2}{*}{Agent} & \multicolumn{4}{c|}{$Q$}\\\cline{2-5}
             & $1_F1_B$ & $1_F2_B$ & $2_F1_B$ & $2_F2_B$ \\\hline
            1 &  $v$    & $w$   & $y$     & $z$ \\
            2 &  $z$    & $y$   & $w$     & $v$ \\ \hline
        \end{tabular}
    \end{center}
	\end{minipage}
	\end{minipage}\vspace{1em}


	Here $v,w,y,z$ are probabilities of these four discrete assignments and there exists $v+w+y+z=1$.
	Due to \sdef, we have following inequalities in terms of agent $1$: It is easy to see that
	$\sum_{{\vx} \in \ucs(\succ_1,2_F1_B)}{q_{1,\vx}}=1=\sum_{{\vx} \in \ucs(\succ_1,2_F1_B)}{q_{2,\vx}}$. In addition,
	\begin{displaymath}
	\begin{split}
	\sum_{{\vx} \in \ucs(\succ_1,1_F1_B)}{q_{1,\vx}}=&v \geq z=\sum_{{\vx} \in \ucs(\succ_1,1_F1_B))}{q_{2,\vx}} \\
	\sum_{{\vx} \in \ucs(\succ_1,1_F2_B)}{q_{1,\vx}}=&v+w\geq z+y=\sum_{{\vx} \in \ucs(\succ_1,1_F2_B)}{q_{2,\vx}}\\
	\sum_{{\vx} \in \ucs(\succ_1,2_F2_B)}{q_{1,\vx}}=&v+w+z\geq z+y+v=\sum_{{\vx} \in \ucs(\succ_1,2_F2_B)}{q_{2,\vx}}
	\end{split}
	\end{displaymath}
	Similarly, for agent $2$, there exists $y \geq w, y+w+z\geq w+y+v$.
	Thus $w=y,v=z,v+w=y+z=0.5$.
	Because $Q$ is \sdopta, $P\nsd{\null}Q$.
	Suppose that $P\nsd{1}Q$.
	Therefore at least one of the following inequalities is true:
	\begin{equation} \label{prfir2}
	\begin{split}
	\sum_{{\vx} \in \ucs(\succ_1,1_F1_B)}{q_{1,\vx}}=&v> 0.5 =  \sum_{{\vx} \in \ucs(\succ_1,1_F1_B)}{p_{1,\vx}} \\
	\sum_{{\vx} \in \ucs(\succ_1,1_F2_B)}{q_{1,\vx}}=&v+w> 0.5 = \sum_{{\vx} \in \ucs(\succ_1,1_F2_B)}{p_{1,\vx}} \\
	\sum_{{\vx} \in \ucs(\succ_1,2_F2_B)}{q_{1,\vx}}=&v+w+z> 1 = \sum_{{\vx} \in \ucs(\succ_1,2_F2_B)}{p_{1,\vx}} \\
	\sum_{{\vx} \in \ucs(\succ_1,2_F1_B)}{q_{1,\vx}}=&v+w+z+y> 1 = \sum_{{\vx} \in \ucs(\succ_1,2_F1_B)}{p_{1,\vx}}
	\end{split}
	\end{equation}
	Since $v=z\leq v+w=y+z=0.5$, all inequalities in (\ref{prfir2}) do not hold.
	Thus we have $P\sd{1}Q$.
	With a similar analysis, we can also obtain that $P\sd{2}Q$.
	Therefore we have that $P\sd{\null}Q$ and $P\neq Q$, which is contradictory to the assumption.\qed
\end{proof}


\begin{remark}\label{rm:tightimp}
    \Cref{thm:imp} can be tightened under LP-tree preferences \cite{Booth10:Learning,Sikdar18:Top} contained in the strict linear preferences, with \sdwopt{} \cite{hashimoto2011characterizations,Bogomolnaia12:Probabilistic} implied by \sdopt{},  and \sdwef{} \cite{Bogomolnaia01:New} implied by \sdef{}.

    An LP-tree preference profile is that each agent's preference can be represented as a rooted directed tree with each node labeled by type and a conditional preference table.
    In each path from root to a leaf, every type occurs only once.
    For a node labeled by type $t$, its conditional preference table is a strict linear order over items in $D_t$, and each edge from it is labeled by each of these items.
    The \sdwopt{} only concerns that any two agents cannot improve their allocations by exchanging shares between them.
    The \sdwef{} requires that for any agent, the other agents' allocations are not better than hers.
    Here we give the brief definition of these two properties, and the proof of the tighter result is in \Cref{sec:rm:tightimp}.

    {\bf \sdwopt{}:} The mechanism $f$ satisfies \sdwopt{} if for any $R\in\mr$, there is no \ram{} assignment $P\neq f(R)$ such that $P\sd{j}f(R)$ for every $j\le n$ and $|\{j\in N: P_j\neq f(R)_j\}|\leq2$.

    {\bf \sdwef{}:} The mechanism $f$ satisfies \sdwef{} if for every pair of agents $j,\hat j\le n, P_{\hat j}\sd{j}P_j\Rightarrow P_{\hat j}=P_j$.

    We give the proof in \Cref{sec:rm:tightimp}.
\end{remark}

\section{\mtaps{} with Indivisible Items and Lexicographic Preferences}
\label{LPS}

In this section, we first introduce the lexicographic preference, and then develop \lps{} as a specialized mechanism for \mtaps{} where the items are indivisible and agents' preferences are restricted to the lexicographic preference domain, and show that it retains the good properties of PS.

Faced with the impossibility result of \Cref{thm:imp} we wonder if it is possible to circumvent it by adding some reasonable restriction.
We finally choose lexicographic preferences as a restriction on the preference domain in our setting according to the result in \Cref{rm:tightimp}.
An agent with a lexicographic preference compares two bundles by comparing the items of each type in the two bundles one by one according to the importance order on types. 
In short, the agent takes the importance of types into consideration while ranking bundles.
We say that a certain \mtap{} is under the restriction of lexicographic preferences if it satisfies that for every agent $j\le n$, $\succ_j$ is a lexicographic preference.
We give its formal definition in the following with the notation $\type{i}{\vx}$: Given any bundle $\vx\in\md$ and any type $i\le p$, $\type{i}{\vx}$ refers to the item of type $i$ in $\vx$.

\begin{definition}\label{dfn:lex}
	(\textbf{lexicographic preference}) Given an \mtap{}, the preference relation $\succ_j$ of agent $j\le n$ is lexicographic if there exists
	\begin{enumerate*}[label=(\roman*)]
		\item importance order i.e. a strict linear order $\impord{j}$ over $p$ types,
		\item for each type $i\le p$, a strict linear order $\prefintype{j}{i}$ over $D_i$,
	\end{enumerate*}
	 such that for every pair $\vx,\hat\vx \in \md$, $\vx \succ_j \hat\vx$ if and only if there exists a type $i$ s.t. $\type{i}{\vx} \prefintype{j}{i} \type{i}{\hat\vx}$, and for all $\hat i \impord{j} i$, $\type{\hat i}{\vx}=\type{\hat i}{\hat\vx}$.
\end{definition}

For example, the preference $1_F2_B \succ 1_F1_B \succ 2_F2_B \succ 2_F1_B$ is a lexicographic preference represented as importance order $F\impord{\null} B$, and strict linear orders $1_F \prefintype{\null}{F} 2_F$ and $2_B \prefintype{\null}{B} 1_B$ for two types.
We also note that although lexicographic preference and lexicographic dominance looks similar, lexicographic preference is used in ranking bundles in agents' preferences, while lexicographic dominance is used to compare allocations or assignments consisting of shares of bundles.

\subsection{\textbf{Algorithm for \lps{}}}

Before going any further with \lps{}, we introduce some notations for ease of exposition.
We use $\srd{P}{i}$ to denote the \ram{} assignment of items of type $i$ at $P$. $\srd{P}{i}$ is a $|N| \times |D_i| $ matrix and for any $i\le p$, any $o\in D_i$, $\srd{p}{i}_{j,o}=\sum_{o\in {\vx}, {\vx}\in \md}p_{i,\vx}$ represents the total share of bundles containing $o$ consumed by agent $j$.
To distinguish from single type \ram{} assignments, we call the ones for \mtaps{} as multi-type \ram{} assignments.
Besides, for $o\in D_i$, we use the notation of the upper contour set $\ucs(\prefintype{}{i},o)$ to refer to items of type $i$ that better or equal to $o$ regarding $\prefintype{}{i}$.

\begin{algorithm}
	\begin{algorithmic}[1]
		\State {\bf Input:} An \mtap{} $(N,M)$, a lexicographic preference profile $R$.
		\State For every $o\in M$, $\s(o)\gets 1$. For every $i \leq n$, $\srd{P}{i}\gets 0^{n\times n}$. $P\gets 0^{n\times|\md|}$.
        \Loop\ $p$ times
            \State {\bf Identify top type} $i_j$ for every agent $j\le n$.

            //For the $k$th loop, $i_j$ is the $k$th type regarding $\impord{j}$.
            \For{$i\le p$}
                \State $t \gets 0$. 
                \State $N^i=\{j\le n|i_j=i\}$.
        		\While{$t<1$}
        		
            		\State {\bf Identify top item} $\topb{j}{i}$ in type $i$ for every agent ${j\in N^i}$.
            		\parState{ {\bf Consume.}
            			\begin{itemize}[topsep=0em,partopsep=0em]
            				\item[10.1:] For each $o\in D_i$,${\con{o}{\null}\gets|\{j\in N^i:\topb{j}{i}=o\}|}$.
            				\item[10.2:] $\prog{\null}\gets \min_{o\in D_i}\frac{\s(o)}{\con{o}{\null}}$.
            				\item[10.3:] For each $j\in N^0$, $\srd{p}{i}_{j,\topb{j}{i}}\gets \srd{p}{i}_{j,\topb{j}{i}}+\prog{\null}$.
            				\item[10.4:] For each $o\in D_i$, $\s(o)\gets\s(o)-{\prog{\null}\times \con{o}{\null}}$.
            				\item[10.5:] $t \gets t + \prog{\null}$.
            			\end{itemize}}
        		\EndWhile
            \EndFor
        \EndLoop\vspace{-1em}
        \State For every $j\leq n,\vx\in\md$,$p_{j,\vx}=\prod_{o=\type{i}{\vx},i\leq p}{\srd{p}{i}_{j,o}}$.
		\State \Return $P$
	\end{algorithmic}
	\caption{\label{alg:lps} \lps{}}
\end{algorithm}

In the \lps{} mechanism, agent $j$ always consumes her favorite item $o_j$ available in the current most important type.
Agent $j$ would not stop consuming $o_j$ unless one of the following occurs:
\begin{enumerate}[label=(\roman*),itemindent=2em,topsep=0em]
  \item $o_j$ is exhausted, and then agent $j$ will consume her next favorite and available item according to $\prefintype{j}{i}$;
  \item $\sum_{o\in D_i}{\srd{p}{i}_{j,o}}=1,o_j \in D_i$, and then agent $j$ will turn to her next most important type $\hat i$ according to $\impord{j}$ and consume her favorite item available in that type.
\end{enumerate}
\hspace{1.5em}After consumption, we obtain $\srd{P}{i}$ for every type $i \leq p$.
With the assumption that allocations among different types are independent, we construct $P$ by
\begin{equation}\label{eq:lps:share}
    p_{j,\vx}=\prod_{o=\type{i}{\vx},i\leq p}{\srd{p}{i}_{j,o}}.
\end{equation}

We divide the execution of \lps{} into $p$ phase where each agent only consume items in her current most important type $i_j$, and the time $t$ for each phase is up to one unit.
In the beginning of each phase, we set the timer $t=0$.
During the consumption, agent $j$ first decides her most preferred {\em unexhausted} item $\topb{j}{i}$ in her current most importance type $i$ regarding $\prefintype{j}{i}$.
Here we say an item $o$ is exhausted if the supply $\s(o)=0$.
Each agent consumes the item at a uniform rate of one unit per unit of time.
The consumption pauses whenever one of the items being consumed becomes exhausted.
That means agent $j$'s share on $\topb{j}{i}$ is increased by $\prog{\null}$,
the duration since last pause, and the supply $\s(o)$ is computed by subtracting $\prog{\null}$ for $\con{o}{\null}$ times, the number of agent $j$ which satisfies $\topb{j}{i}=o$.
In \Cref{alg:lps}, $\prog{\null}$ is computed by $\min_{o\in M}\frac{\s(o)}{\con{o}{\null}}$.
After that we increase the timer $t$ by $\prog{\null}$, identify $\topb{j}{i}$ for each agent, and continue the consumption.
The current phase ends when the timer $t$ reaches $1$, and the algorithm starts the next phase.
We show that \lps{} is able to deal with indivisible items while retaining the good properties in \Cref{thm:lps}.



\begin{figure}[ht]
	\centering
	\includegraphics[width=0.8\linewidth]{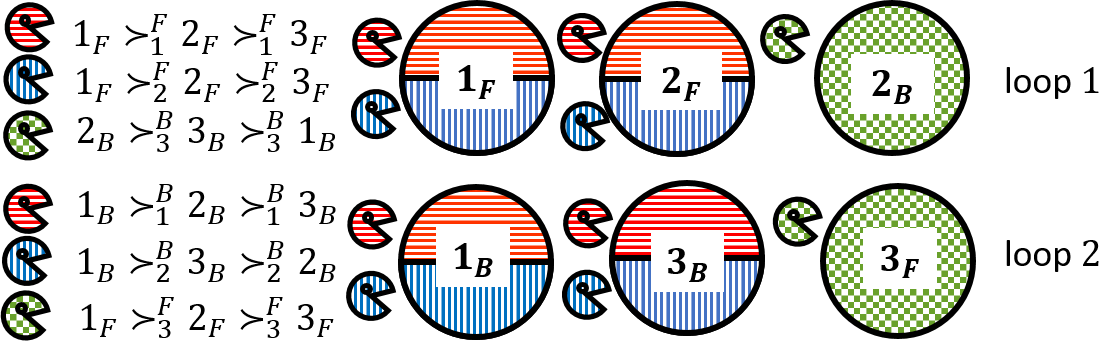}
	\caption{Execution of \lps{} in \Cref{eg:lps}.}
	\label{fig:lps}
\end{figure}



\begin{example}\label{eg:lps}
    Consider an \mtap{} where $N=\{1,2,3\}$, $M=D_F\times D_B, D_F=\{1_F,2_F,3_F\}$, $D_B=\{1_B,2_B,3_B\}$, and the profile $R=\{\succ_1,\succ_2,\succ_3\}$. The preferences $\succ_1,\succ_2,\succ_3$ are as follows:
    \begin{center}
        \begin{tabular}{|c|c|}\hline
            Agent & Preferences \\\hline
            1 & $F\impord{1}B, 1_F\prefintype{1}{F}2_F\prefintype{1}{F}3_F,1_B\prefintype{1}{B}2_B\prefintype{1}{B}3_B$ \\
            2 & $F\impord{2}B,  1_F\prefintype{2}{F}2_F\prefintype{2}{F}3_F,1_B\prefintype{2}{B}3_B\prefintype{2}{B}2_B$ \\
            3 & $B\impord{3}F,  1_F\prefintype{3}{F}2_F\prefintype{3}{F}3_F,2_B\prefintype{3}{B}3_B\prefintype{3}{B}1_B$\\\hline
        \end{tabular}
    \end{center}
	The execution of \lps{} is shown in {\Cref{fig:lps}}.
	In loop $1$, agent $1$ and $2$ consume items in $D_F$, while agent $3$ consumes alone in $D_B$.
	Therefore agent $3$ gets her favorite items $2_B$ in $D_B$ fully, and $1_B$ and $3_B$ are left.
	Since agent $1$ and $2$ have the same preference for $D_F$, they each obtain $0.5$ units of $1_F$ and $0.5$ units of $2_F$, and $3_F$ is left.
	Similarly in loop $2$, agent $1$ and $2$ prefers type $B$ while agent $3$ prefers $F$.
	Then agent $3$ gets the remaining item $3_F$, and agent $1$ and $2$ divide $1_B$ and $3_B$ uniformly according to their preferences.
    \begin{center}
        \centering
        \begin{tabular}{|c|ccccc|}\hline
            \multirow{2}{*}{Agent} & \multicolumn{5}{c|}{$P$}\\\cline{2-6}
             & $1_F1_B$ & $1_F3_B$ & $2_F1_B$ & $2_F3_B$ & $3_F2_B$\\\hline
            1 &  0.25    & 0.25   & 0.25     & 0.25       & 0       \\
            2 &  0.25    & 0.25   & 0.25     & 0.25       & 0       \\
            3 &  0      & 0       & 0        & 0          & 1       \\\hline
        \end{tabular}
    \end{center}
	The consumption above results in the final assignment $P$, and it is easy to check that $P$ is \deca{}.
\end{example}

\subsection{\textbf{Properties}}

\Cref{thm:lps} shows that, as an extension of PS, \lps{} inherits efficiency and envyfreeness based on stochastic dominance in solving \mtaps{} with lexicographic preferences, and it is able to deal with indivisible items for its \deca{} outputs.
\begin{theorem}\label{thm:lps}
    For \mtaps{} with lexicographic preferences, \lps{} satisfies \sdopt{} and \sdef{}.
    Especially, \lps{} outputs \deca{} assignments.
\end{theorem}
\begin{proof}
    Given a \mtap{}$(N,M)$, and profile $R$ of lexicographic preferences, let $P=\lps(R)$.
    In the following proof, we use $\itemwt{i}$ to refer to an arbitrary item of type $i$, and $(\itemwt{i}*)$ to refer to any bundle containing $\itemwt{i}$.

    (1) (\textbf{\sdopt})
    By contradiction, we suppose in another assignment $Q$ there exists agent $j$ who get a better allocation than in $P$ when others are also better or not affected, which means $Q\sd{}P$.
    W.l.o.g, label the type according to $\impord{j}$ as $1 \impord{j} 2 \impord{j} \dots \impord{j} p$.
    We show that $Q_j=P_j$ by proving the following equation with mathematical induction: for any $i\leq p$ and any $\itemwt{1},\dots,\itemwt{i}$,
    \begin{equation}\label{eq:thm:lps:1}
        \sum_{\vx=(\itemwt{1},\dots,\itemwt{i}*)}{p_{j,\vx}}=\sum_{\vx=(\itemwt{1},\dots,\itemwt{i}*)}{q_{j,\vx}}.
    \end{equation}

    First, we prove (\ref{eq:thm:lps:1}) when $i=1$ i.e. $\srd{Q}{1}_j=\srd{P}{1}_j$.
    If $\srd{Q}{1}\sd{j}\srd{P}{1}$ is false, we have that there exists $\hitemwt{1}$ and the least preferred bundle $\hat\vx$ containing $\hitemwt{1}$ regarding $\succ_j$ such that
    \begin{equation*}
        \sum_{\vx\in\ucs(\succ_j,\hat\vx)}{p_{j,\vx}}=
        \sum_{o\in\ucs(\prefintype{j}{1},\hitemwt{1})}{p_{j,o}}>
        \sum_{o\in\ucs(\prefintype{j}{1},\hitemwt{1})}{q_{j,o}}=
        \sum_{\vx\in\ucs(\succ_j,\hat\vx)}{q_{j,\vx}}.
    \end{equation*}
    That is a contradiction to the assumption $Q\sd{j}P$.
    Therefore we suppose $\srd{Q}{1}\sd{j}\srd{P}{1}$ and $\srd{Q}{1}_j\neq\srd{P}{1}_j$.
    In some way, we say $\srd{P}{1}$ can be improved to $\srd{Q}{1}$ by shares transferring of bundles.
    Let $N_1$ denote the set of agents who consume items of type~$1$ in Phase~$1$.
    Let $\hitemwt{1}$ be the least preferred item agent $j$ gets in $\srd{P}{1}$ according to $\prefintype{j}{1}$.
    From \Cref{alg:lps}, we learn that agents in $N_1$ obey the rule of PS in Phase~$1$, and therefore the partial assignment for them $\srd{P}{1}_{N_1}=(\srd{P}{1}_j)_{j\in N_1}$ is \sdopta{}.
    Besides, agents not in $N_1$ only gets bundles $\vx$ with $\type{1}{\vx}$ not better than $\hitemwt{1}$.
    By assumption there exists $\itemwt{1}$ such that $\srd{p}{1}_{j,\itemwt{1}}<\srd{q}{1}_{j,\itemwt{1}}$ and $\sum_{o\prefintype{j}{1}\itemwt{1}}{\srd{p}{1}_{j,o}}=
    \sum_{o\prefintype{j}{1}\itemwt{1}}{\srd{q}{1}_{j,o}}$, and therefore we have that the share transferring:
    \begin{enumerate}[label=(\roman*),itemindent=2em,topsep=0em]
      \item only involves agents in $N_1$, which is a contradiction of \sdopt{} of $\srd{P}{1}_{N_1}$;
      \item involves agents not in $N_1$ and the best items of type $1$ w.r.t. $\prefintype{j}{1}$ they may have is $\hitemwt{1}$.
      If $\itemwt{1}\neq\hitemwt{1}$, the extra share of $\srd{q}{1}_{j,\itemwt{1}}$ comes from agents in $N_1$, which is a contradiction of \sdopt{} as in Case (\romannumeral1).
      Therefore $\itemwt{1}=\hitemwt{1}$ and we have $\sum_{o\in D_1}{\srd{q}{1}_{j,o}}=
      \sum_{o\in\ucs(\prefintype{j}{1},\itemwt{1})}{\srd{q}{1}_{j,o}}>
      \sum_{o\in\ucs(\prefintype{j}{1},\itemwt{1})}{\srd{p}{1}_{j,o}}=
      \sum_{o\in D_1}{\srd{p}{1}_{j,o}}=1$, which is a contradiction.
    \end{enumerate}
    Therefore we have $\srd{Q}{1}_j=\srd{P}{1}_j$.

    Then we prove (\ref{eq:thm:lps:1}) for type $i>1$ while for any $\itemwt{1},\dots,\itemwt{i-1}$,
    \begin{equation}\label{eq:thm:lps:2}
        \sum_{\vx=(\itemwt{1},\dots,\itemwt{i-1}*)}{p_{j,\vx}}=
        \sum_{\vx=(\itemwt{1},\dots,\itemwt{i-1}*)}{q_{j,\vx}}.
    \end{equation}
    It is easy to see that here (\ref{eq:thm:lps:1}) is equivalent to $\srd{Q}{i}=\srd{P}{i}$.
    We first show that like in type $1$ it is necessary that $\srd{Q}{i}\sd{j}\srd{P}{i}$.
    For any $\hitemwt{1},\dots,\hitemwt{i}$, let $\hat\vx$ be the least preferred bundle containing them regarding $\succ_j$.
    Let
    \begin{equation*}
    \begin{split}
    \md_i=&\{(\hitemwt{1},\dots,\hitemwt{i-1},\itemwt{i}*)|\itemwt{i}\in\ucs(\prefintype{j}{i},\hitemwt{i})\}, \\
    \md'_i=&\{(\itemwt{1},\dots,\itemwt{i-1}*)|\itemwt{1}\in\ucs(\prefintype{j}{1},\hitemwt{1}),\dots,\itemwt{i-2}\in\ucs(\prefintype{j}{i-2},\hitemwt{i-2}),\itemwt{i-1}\prefintype{j}{i-1}\hitemwt{i-1}\}.
    \end{split}
    \end{equation*}
    By assumption $Q\sd{j}P$ we have that
    \begin{equation}\label{eq:thm:lps:3}
    \begin{split}
        \sum_{\vx\in\md_i}{p_{j,\vx}}+\sum_{\vx\in\md'_i}{p_{j,\vx}}=
        \sum_{\vx\in\ucs(\succ_j,\hat\vx)}{p_{j,\vx}}
        \leq \sum_{\vx\in\ucs(\succ_j,\hat\vx)}{q_{j,\vx}}=
        \sum_{\vx\in\md_i}{q_{j,\vx}}+\sum_{\vx\in\md'_i}{q_{j,\vx}}
    \end{split}
    \end{equation}
    With (\ref{eq:thm:lps:2}) we have $\sum_{\vx\in\md'_i}{p_{j,\vx}}=\sum_{\vx\in\md'_i}{q_{j,\vx}}$, and induce from (\ref{eq:thm:lps:3}) that $\sum_{\vx\in\md_i}{p_{j,\vx}}\leq \sum_{\vx\in\md_i}{q_{j,\vx}}$.
    By summing up each side by $\hitemwt{1},\dots,\hitemwt{i-1}$, we have that:
    \begin{equation*}
        \sum_{\hitemwt{1}}\dots\sum_{\hitemwt{i-1}}\sum_{\vx\in\md_i}{p_{j,\vx}}\leq \sum_{\hitemwt{1}}\dots\sum_{\hitemwt{i-1}}\sum_{\vx\in\md_i}{q_{j,\vx}}.
    \end{equation*}
    According to the definition of lexicographic preference, it is equal to:
    \begin{equation*}
        \sum_{\vx\in\{(\itemwt{i}*)|\itemwt{i}\in\ucs(\prefintype{j}{i},\hitemwt{i})\}}{p_{j,\vx}}\leq \sum_{\vx\in\{(\itemwt{i}*)|\itemwt{i}\in\ucs(\prefintype{j}{i},\hitemwt{i})\}}{q_{j,\vx}}.
    \end{equation*}
    It can also be written as:
    \begin{equation*}
        \sum_{\itemwt{i}\in\ucs(\prefintype{j}{i},\hitemwt{i})}{\srd{p}{i}_{j,\itemwt{i}}}\leq \sum_{\itemwt{i}\in\ucs(\prefintype{j}{i},\hitemwt{i})}{\srd{q}{i}_{j,\itemwt{i}}}
    \end{equation*}
    That means $\srd{Q}{i}\sd{j}\srd{P}{i}$.

    With $\srd{Q}{i}\sd{j}\srd{P}{i}$, next we prove $\srd{Q}{i}\neq\srd{P}{i}$ is false by contradiction as for type $i$.
    Let $N_i$ denote the set of agents who consume items of type~$i$ in Phase~$i$ and $N'_i$ denote agents consume items of type~$i$ after Phase~$i$.
    Notice that agents which have obtained items of type $i$ before Phase $i$ are not considered because they can never benefit by trading shares with agents in $N_i\bigcup N'_i$.
    Let $\hitemwt{i}$ be the least preferred item agent $j$ gets in $\srd{P}{i}$ according to $\prefintype{j}{i}$.
    As we have shown in the proof for type $1$, $\srd{P}{i}_{N_i}=(\srd{P}{i}_j)_{j\in N_i}$, the partial assignment for agents in $N_i$, is \sdopta{} regarding the available items of type $i$ in Phase $i$, and agents in ${N'_i}$ only gets bundles $\vx$ with $\type{i}{\vx}$ not better than $\hitemwt{i}$.
    By assumption there exists $\itemwt{i}$ such that
    $\srd{p}{i}_{j,\itemwt{i}}<\srd{q}{i}_{j,\itemwt{i}}$ and $\sum_{o\prefintype{j}{i}\itemwt{i}}{\srd{p}{i}_{j,o}}=
    \sum_{o\prefintype{j}{i}\itemwt{i}}{\srd{q}{i}_{j,o}}$,
    and therefore we have that the share transferring:
    \begin{enumerate}[label=(\roman*'),itemindent=2em,topsep=0em]
      \item only involves agents in $N_i$, which is a contradiction of \sdopt{} of $\srd{P}{i}_{N_i}$;
      \item involves agents in $N'_i$, and the best items of type $i$ w.r.t. $\prefintype{j}{i}$ they may have is $\hitemwt{i}$.
      If $\itemwt{i}\neq\hitemwt{i}$, the extra share of $\srd{q}{i}_{j,\itemwt{i}}$ comes from agents in $N_i$, which is a contradiction of \sdopt{} as in Case (\romannumeral1').
      Therefore $\itemwt{i}=\hitemwt{i}$ and we have  $\sum_{o\in D_i}{\srd{q}{i}_{j,o}}=
      \sum_{o\in\ucs(\prefintype{j}{i},\itemwt{i})}{\srd{p}{i}_{j,o}}>
      \sum_{o\in\ucs(\prefintype{j}{i},\itemwt{i})}{\srd{q}{i}_{j,o}}=
      \sum_{o\in D_i}{\srd{q}{i}_{j,o}}=1$, which is a contradiction.
    \end{enumerate}
    Therefore we have the result.

    By mathematical induction, we have that $p_{j,\vx}=q_{j,\vx}$, which is a contradiction to $Q_j\neq P_j$.


    (2) (\textbf{\sdef})
    As agents spend one unit of time for each type, we divide the execution of \lps{} into $p$ phases by type, and the following proof develops by phases.

    We first declare that agent $j$ does not envy other agents who have the same importance order.
    For convenience, we label the types according to $\impord{j}$.
    Let $N_i$ be the set of agents who consume items of types $i$ in Phase $i$.
    Phase $i$ can be viewed as the execution of PS for the single type allocation problem with agents in $N_i$ and available items left in $D_i$ at Phase $i$.
    By \cite{Bogomolnaia01:New}, we know that PS satisfies \sdef.
    Therefore we have that for any $k \in N_i$, $\sum_{o'\in\ucs(\prefintype{j}{i},o)}{\srd{p}{i}_{j,o'}} \geq \sum_{o'\in\ucs(\prefintype{j}{i},o)}{\srd{p}{i}_{k,o'}}$ for any $o \in D_i$, which also means $\sum_{\vx'\in\ucs(\succ_j,\vx)}{p_{j,\vx'}} \geq \sum_{\vx'\in\ucs(\succ_j,\vx)}{p_{k,\vx'}}$ for any $\vx\in\md$.

    We next declare that agent $j$ does not envy agents who have different important orders.
    Assume by contradiction that $k$ is such an agent and there exists ${\hat\vx} \in \md$ which satisfies
    \begin{equation}\label{eq:lpsensup}
        \sum_{\vx \in \ucs(\succ_j,\hat\vx)}p_{j,\vx} < \sum_{\vx \in \ucs(\succ_j,\hat\vx)}p_{k,\vx}.
    \end{equation}
    We prove the declaration by mathematical induction.
    First we prove agent $k$'s most important type is $1$.
    If not, then agent $k$ would consume $o \in D_1$ in the later phase.
    However, the items left in $D_1$ at phase $i\neq1$ are not more preferable than those consumed by agent $j$.
    Thus for any $\hitemwt{1},\itemwt{1}$ satisfying $\srd{p}{1}_{j,\hitemwt{1}}>0,\srd{p}{1}_{k,\itemwt{1}}>0$, we have $\hitemwt{1}\prefintype{i}{1}{\itemwt{1}}$ or $\hitemwt{1}=\itemwt{1}$.
    We discuss them respectively.

    Case (\romannumeral1):
    If $\hitemwt{1} \prefintype{j}{1} \itemwt{1}$ for any $\hitemwt{1},\hitemwt{1}$ satisfying $\srd{p}{1}_{j,\itemwt{1}}>0,\srd{p}{1}_{k,\itemwt{1}}>0$.
    then we obtain that for any $\hat\vx,\vx,p_{j,\hat\vx}>0,p_{k,\vx}>0$, we have ${\hat\vx}\succ_j{\vx}$  since $\type{1}{\hat\vx} \prefintype{j}{1} \type{1}{\vx}$.
    Let $\hat\vx$ be least preferable bundle according to $\succ_j$ and $p_{j,\hat\vx}>0$, and we have that for any $\vx'\in\ucs(\succ_j,\hat\vx)$,
    \begin{equation}\label{eq:lpsef1}
        \sum_{\vx\in \ucs(\succ_j,\vx')}p_{j,\vx} \geq \sum_{\vx\in \ucs(\succ_j,\vx')}p_{k,\vx} = 0
    \end{equation}
    For any $\vx'$ which satisfies $\hat\vx\succ_j\vx'$,
    \begin{equation}\label{eq:lpsef2}
        \sum_{\vx \in \ucs(\succ_j,\vx')}p_{j,\vx} =1 \geq \sum_{\vx \in \ucs(\succ_j,\vx')}p_{k,\vx}
    \end{equation}
    Combining (\ref{eq:lpsef1}) and (\ref{eq:lpsef2}), we obtain that $\sum_{\vx \in \ucs(\succ_j,\vx')}p_{j,\vx}\geq\sum_{\vx \in \ucs(\succ_j,\vx')}p_{k,\vx}$ for any $\vx'$, a contradiction to (\ref{eq:lpsensup}).

    Case (\romannumeral2):
    If there exists $\hitemwt{1}$ satisfying $\srd{p}{1}_{j,\hitemwt{1}}>0,\srd{p}{1}_{k,\hitemwt{1}}>0$, then $\hitemwt{1}$ will be the least preferable item consumed by agent $j$, and the most preferable one consumed by agent $k$ both regarding $\prefintype{j}{1}$.
    We obtain that
    \begin{equation}\label{eq:lpsef3}
        \sum_{o\prefintype{j}{1}\hitemwt{1}}{\srd{p}{1}_{j,\itemwt{1}}}
        = 1 - \srd{p}{1}_{j,\hitemwt{1}}
        \geq \srd{p}{1}_{k,\hitemwt{1}} =
        \sum_{o \in \ucs(\prefintype{j}{1},\hitemwt{1})}{\srd{p}{1}_{k,\itemwt{1}}}.
    \end{equation}
    It means that for every $\hat\vx=(\hitemwt{1}*)$,
    \begin{equation*}\label{}
        \sum_{\vx \in \ucs(\succ_j,\hat\vx)}p_{j,\vx}
        \geq \sum_{o\prefintype{j}{1}\hitemwt{1}}{\srd{p}{1}_{j,o}}
        \geq \sum_{o \in \ucs(\prefintype{j}{1},\hitemwt{1})}{\srd{p}{1}_{k,o}}
        \geq \sum_{\vx \in \ucs(\succ_j,\hat\vx)}p_{k,\vx}
    \end{equation*}
    When $\hat\vx\neq(\hitemwt{1}*)$, we have (\ref{eq:lpsef1}) for $\type{1}{\hat\vx}\prefintype{j}{1}\itemwt{1}$ and (\ref{eq:lpsef2}) for $\itemwt{1}\prefintype{j}{i}\type{1}{\hat\vx}$, which cover the remaining cases.
    Therefore agent $k$'s first important type is $1$.

    Then we prove the agent $k$'s $i$th important type is $i$ when her ${i'}$th type is $\hat i$ for $i'<i$, i.e. $\srd{p}{i'}_{j,\itemwt{i'}}=\srd{p}{i'}_{k,\itemwt{i'}}$ for $i'<i$.
    Suppose it is false and we discuss by cases above as for type $1$.
    Because agent $k$ consumes items in type $i$ later that agent $j$, we note that for any $\itemwt{i},\hitemwt{i}$ satisfying $\srd{p}{i}_{j,\itemwt{i}}>0,\srd{p}{i}_{k,\hitemwt{i}}>0$, we have $\itemwt{i}\prefintype{j}{i}\hitemwt{i}$ or $\itemwt{i}=\hitemwt{i}$.

    Case (\romannumeral1'):
    Suppose for any $\itemwt{i},\hitemwt{i}$ satisfying $\srd{p}{i}_{j,\itemwt{i}}>0,\srd{p}{i}_{k,\hitemwt{i}}>0$ and $\itemwt{i}\neq\hitemwt{i}$, we have $\itemwt{i} \prefintype{j}{i} \hitemwt{i}$.
    That means for $\vx=(\itemwt{1},\dots,\itemwt{i-1},\itemwt{i}*),\hat\vx=(\itemwt{1},\dots,\itemwt{i-1},\hitemwt{i}*)$ satisfying $p_{j,\vx}>0,p_{k,\hat\vx}>0$, we have that $\vx\succ_j\hat\vx$. 
    For any $\pitemwt{1},\dots,\pitemwt{i}$, let $\vx'=(\pitemwt{1},\dots,\pitemwt{i}*)$.
    Then we can obtain the result similar to
    (\ref{eq:lpsef1}) and (\ref{eq:lpsef2})
    in Case~(\romannumeral1) regarding (\ref{eq:lps:share}), the computation way of bundle shares in \lps{}:
    \begin{equation}\label{eq:lpsef7}
        \sum_{\vx=(\pitemwt{1},\dots,\pitemwt{i-1}*), \vx\in\ucs(\succ_j,\vx')}{p_{j,{\vx}}}
        \geq
        \sum_{\vx=(\pitemwt{1},\dots,\pitemwt{i-1}*), \vx\in\ucs(\succ_j,\vx')}{p_{k,{\vx}}}
    \end{equation}
    In addition, we can take the cumulative shares of upper contour set apart as follows:
    \begin{equation}\label{eq:lpsef4}
	\begin{split}
	    \sum_{{\vx}\in \ucs(\succ_j,\vx')}{p_{j,{\vx}}} =&
            \sum_{{\vx}\in\{(\itemwt{1}*)|\itemwt{1} \prefintype{j}{1} \pitemwt{1}\}}{p_{j,{\vx}}} +
            \sum_{{\vx}\in\{(\pitemwt{1},\itemwt{2}*)|\itemwt{2} \prefintype{j}{2} \pitemwt{2}\}}{p_{j,{\vx}}} \\
            &+\dots+
            \sum_{\vx=(\pitemwt{1},\dots,\pitemwt{i-1}*), \vx\in\ucs(\succ_j,\vx')}{p_{j,{\vx}}}
	\end{split}
	\end{equation}
	By the computation of shares for bundles (\ref{eq:lps:share}) and for any $\hat{i}<i$, $\srd{p}{\hat i}_{j,\itemwt{\hat i}}=\srd{p}{\hat i}_{k,\itemwt{\hat i}}$, we have that
	\begin{equation}\label{eq:lpsef5}
	    \sum_{{\vx}\in\{(\pitemwt{1},\dots,\itemwt{\hat{i}}*)|\itemwt{\hat{i}} \prefintype{j}{\hat{i}} \pitemwt{\hat{i}}\}}{p_{j,{\vx}}}=
	    \sum_{{\vx}\in\{(\pitemwt{1},\dots,\itemwt{\hat{i}}*)|\itemwt{\hat{i}} \prefintype{j}{\hat{i}} \pitemwt{\hat{i}}\}}{p_{k,{\vx}}}
	\end{equation}
	Therefore with (\ref{eq:lpsef7}), (\ref{eq:lpsef4}) and (\ref{eq:lpsef5}) we have $\sum_{{\vx}\in  \ucs(\succ_j,\vx')}{p_{j,{\vx}}}\geq\sum_{{\vx}\in \ucs(\succ_j,\vx')}{p_{k,{\vx}}}$ for any $\vx'$, a contradiction to (\ref{eq:lpsensup}).

    Case (\romannumeral2'):
    Suppose there exists $\hitemwt{i}$ satisfying $\srd{p}{i}_{j,\hitemwt{i}}>0,\srd{p}{i}_{k,\hitemwt{i}}>0$.
    Therefore $\hitemwt{i}$ will be the least preferable item consumed by agent $j$, and the most preferable one consumed by agent $k$ both regarding $\prefintype{j}{i}$.
    Let $\vx'=(\pitemwt{1},\dots,\pitemwt{i-1},\pitemwt{i}*)$. 
    We first consider $\pitemwt{i}=\hitemwt{i}$.
    Similar to (\ref{eq:lpsef3}), we have that
    \begin{equation}\label{eq:lpsef3.5}
        \sum_{o\prefintype{j}{i}\hitemwt{i}}{\srd{p}{i}_{j,o}}
        = 1 - \srd{p}{i}_{j,\hitemwt{i}}
        \geq  \srd{p}{i}_{k,\hitemwt{i}} =
        \sum_{o \in \ucs(\prefintype{j}{i},\hitemwt{i})}{\srd{p}{i}_{k,o}}
    \end{equation}
   With computation way of bundle shares (\ref{eq:lps:share}), we induce that for $\alpha=\prod_{i'=1}^{i-1}{\srd{p}{i}_{j,\pitemwt{i'}}}=\prod_{i'=1}^{i-1}{\srd{p}{i}_{k,\pitemwt{i'}}}$,
    \begin{equation}\label{eq:lpsef6}
    \begin{split}
        &\sum_{\vx=(\pitemwt{1},\dots,\pitemwt{i-1}*), \vx\in\ucs(\succ_j,\vx')}{p_{j,{\vx}}}\\
        \geq& \alpha\times\sum_{o\prefintype{j}{i}\hitemwt{i}}{\srd{p}{i}_{j,o}}
        \geq \alpha\times\sum_{o \in \ucs(\prefintype{j}{i},\hitemwt{i})}{\srd{p}{i}_{k,o}}\\
        \geq& \sum_{\vx=(\pitemwt{1},\dots,\pitemwt{i-1}*), \vx\in\ucs(\succ_j,\vx')}{p_{k,{\vx}}}.
    \end{split}
    \end{equation}
    With (\ref{eq:lpsef4}) and (\ref{eq:lpsef5}), (\ref{eq:lpsef6}) means $\sum_{{\vx}\in \ucs(\succ_j,\vx')}{p_{j,{\vx}}}\geq\sum_{{\vx}\in \ucs(\succ_j,\vx')}{p_{k,{\vx}}}$.
    When $\pitemwt{i}\neq\hitemwt{i}$, we have (\ref{eq:lpsef7}) for $\vx'$, and we can obtain the same results as in Case~(\romannumeral1').
    Therefore agent $k$'s first important type is $i$.

    Therefore agent $k$'s important order is the same as agent $j$'s, which is a contradiction.
    With the results above, we prove that agent $j$ envies nobody.

    (3) (\textbf{\deca{} outputs})
    By the classical Birkhoff-Von Neumann theorem, all the single type \ram{} assignments are \deca{}.
    For $i\leq p$, let $\srd{\ma}{i}$ be the set of all the single type deterministic assignments for type $i$, $\srd{A}{i}$ be an arbitrary assignment in $\srd{\ma}{i}$, and we use $\alpha^i$ to denote the possibility of $\srd{A}{i}$.
    Therefore $\srd{P}{i}=\sum{(\alpha_k^i\times\srd{A}{i}_k)}$ for every $i\leq p$.

    We first give a obvious claim that single type deterministic assignments of $p$ types decide an unique multi-type deterministic assignments, and vice versa.
    Specifically, for any multi-type deterministic assignment $A\in\ma$ and $\srd{A}{1},\dots,\srd{A}{p}$ comprising $A$, we have that for every $\vx\in\md$ and $\itemwt{i}=\type{i}{\vx}$,
    \begin{equation}\label{eq:lpsreal0}
        A_{j,\vx}=\prod_{i=1}^{p}{\srd{A}{i}_{j,\itemwt{i}}}.
    \end{equation}
    That means agent $j$ is assigned bundle $\vx$ if she is assigned all the items in $\vx$.
    We show an example of types $\{F,B\}$ for agent $\{1,2\}$ as follows:


    \begin{displaymath}
        \centering
        \begin{split}
        \begin{tabular}{|c|cc|cc|}
            \hline\multirow{2}{*}{Agent} & \multicolumn{2}{c|}{$\srd{P}{F}$} & \multicolumn{2}{c|}{$\srd{P}{B}$}\\\cline{2-5}
             & $1_F$ & $2_F$ & $1_B$ & $2_B$\\\hline
            1 & 1 & 0 & 0 & 1 \\
            2 & 0 & 1 & 1 & 0 \\\hline
        \end{tabular}
        \hspace{1em}
        \begin{tabular}{|c|cccc|}
            \hline\multirow{2}{*}{Agent} & \multicolumn{4}{c|}{$P$}\\\cline{2-5}
             & $1_F1_B$ & $1_F2_B$ & $2_F1_B$ & $2_F2_B$ \\\hline
            1 & 0 & 1 & 0 & 0 \\
            2 & 0 & 0 & 1 & 0 \\\hline
        \end{tabular}
        \end{split}
    \end{displaymath}


    Then by definition of \lps{} we obtain that for any $\vx\in\md$ and any $\itemwt{i}=\type{i}{\vx}$,
    \begin{equation}\label{eq:lpsreal1}
    \begin{split}
       p_{j,\vx} & =\prod_{i\leq p}{\srd{p}{i}_{j,\itemwt{i}}}=\prod_{i\leq p}\sum_{\srd{A}{i}_k\in\srd{\ma}{i}}(\srd{\alpha}{i}_k\times (\srd{A}{i}_k)_{j,\itemwt{i}})
    \end{split}
    \end{equation}
    The result of (\ref{eq:lpsreal1}) is a product of $p$ polynomials, and we can write it as one polynomial and induce as below:
    \begin{equation}\label{eq:lpsreal2}
    \begin{split}
        (\ref{eq:lpsreal1})&=\sum_{\srd{A}{1}_{k_1}\in\srd{\ma}{1},\srd{A}{2}_{k_2}\in\srd{\ma}{2},\dots,\srd{A}{p}_{k_p}\in\srd{\ma}{p}}\prod_{i\leq p}({\alpha^i_{k_i}}\times (\srd{A}{i}_{k_i})_{j,\itemwt{i}}) \\
         & =\sum_{\srd{A}{1}_{k_1}\in\srd{\ma}{1},\srd{A}{2}_{k_2}\in\srd{\ma}{2},\dots,\srd{A}{p}_{k_p}\in\srd{\ma}{p}}(\prod_{i\leq p}{\alpha^i_{k_i}}\times \prod_{i\leq p}(\srd{A}{i}_{k_i})_{j,\itemwt{i}})
    \end{split}
    \end{equation}
    Given $\srd{A}{1}_{k_1},\srd{A}{2}_{k_2},\dots,\srd{A}{p}_{k_p}$, let $A_k$ be the corresponding multi-type assignment consisting of them and $\alpha_k=\prod_{i\leq p}{\srd{\alpha}{i}_{k_i}}$ is the probability of $A_k$ where $\srd{\alpha}{i}_{k_i}$ is the probability of $\srd{A}{i}_{k_i}$ for $\srd{P}{i}$.
    With (\ref{eq:lpsreal0}) we have:
    \begin{equation*}
        (\ref{eq:lpsreal1})=\sum_{A_k\in\ma}(\alpha_k\times (A_k)_{j,\vx})
    \end{equation*}
    Therefore $P$ is \deca{}.
    \qed
\end{proof}

\begin{remark}\label{rm:lpsnlexopt}
\lps{} does not satisfy \lexopt{}.
In \Cref{eg:lps}, agent $1$'s allocation in $P$, denoted $p$ here, is lexicographically dominated by the following allocation $q$:
\begin{center}
    \begin{tabular}{|c|ccccc|}
     \hline Agent & $1_F1_B$ & $1_F3_B$ & $2_F1_B$ & $2_F3_B$ & $3_F2_B$\\\hline
     1     & 0.5      & 0        & 0        & 0.5      & 0 \\\hline
\end{tabular}
\end{center}
We note that $q$ can be obtained by reallocating the shares of items in $p$.
\end{remark}

However, in resistance against manipulation, \lps{} does not perform as well as PS just like most extensions.
\Cref{thm:lpssp} shows that \lps{} is \sdspa{} if the importance orders reported by agents are ensured to be truthful. We also show in \Cref{rm:lpsnsp} how an agent cheats by misreporting her importance order.

\begin{theorem}\label{thm:lpssp}
    For \mtaps{} with lexicographic preferences, \lps{} satisfies \sdsp{} when agents report importance orders truthfully.
\end{theorem}
\begin{proof}
    In the following proof, we use $\itemwt{i}$ to refer to an arbitrary item of type $i$, and $(\itemwt{i}*)$ to refer to any bundle containing $\itemwt{i}$.
	We also use $(\itemwt{i}*)$ as a shorthand to refer to an arbitary bundle containing $\itemwt{i}$.
	
    Suppose agent $j$ misreports its preference as $\succ'_j$ for some types and acquire a better allocation.
    Label the type according to $\impord{j}$ and let $i$ be the most important type for which $j$ misreports.
    Let $P=\lps(R)$, $R'=(\succ'_j,\succ_{-j})$ and $Q=\lps(R')$.
    Then by assumption we have that $Q \sd{j} P,Q\neq P$ and $\srd{p}{\hat i}_{j,o}=\srd{q}{\hat i}_{j,o}$ for every $\hat i < i$.
    The phase when agent $j$ consumes items in $D_i$ can be viewed as executing PS of type $i$.
    By \cite{Bogomolnaia01:New}, PS satisfies \sdsp{}, which means $\srd{Q}{i} \sd{j} \srd{P}{i}$ is false.
    Thus there exists $\hitemwt{i}$ which satisfies
    \begin{equation}\label{eq:lpssp1}
        \sum_{o \in \ucs(\prefintype{j}{i},\hitemwt{i})}\srd{p}{i}_{j,o}>\sum_{o \in \ucs(\prefintype{j}{i},\hitemwt{i})}\srd{q}{i}_{j,o}
    \end{equation}
    With (\ref{eq:lpsef4}) and (\ref{eq:lpsef5}) in \Cref{thm:lps}, we can simplify (\ref{eq:lpssp1}) and obtain that for $\hat\vx=(\hitemwt{1},\dots,\hitemwt{i}*)$,
     \begin{equation*}
            \sum_{{\vx}\in\{(\hitemwt{1},\dots,\hitemwt{i-1},\itemwt{i}*)|\itemwt{i} \in \ucs(\prefintype{j}{i},\hitemwt{i})\}}p_{j,\vx}
            > \sum_{{\vx}\in\{(\hitemwt{1},\dots,\hitemwt{i-1},\itemwt{i}*)|\itemwt{i} \in \ucs(\prefintype{j}{i},\hitemwt{i})\}}q_{j,\vx}.
    \end{equation*}
    It is equal to $\sum_{\vx \in \ucs(\succ_j,\hat\vx)}p_{j,\vx}> \sum_{\vx \in \ucs(\succ_j,\hat\vx)}q_{j,\vx}$, which means agent $j$ does not obtain a better allocation in $Q$, a contradiction.\qed
\end{proof}

\begin{remark}\label{rm:lpsnsp}
    When applying \lps{} to \mtaps{} with lexicographic preferences, an agent may get a better allocation by misreporting her importance order.
\end{remark}
\begin{proof}
    Consider a \mtap{} with lexicographic preferences where there are agents $1,2$ and types $F,B,T$.
    Agents both prefer $1_i$ to $2_i$ for $i \in \{F,B,T\}$, but their preferences over bundle are different due to their importance orders as follow:
    \begin{center}
        \begin{tabular}{|c|c|}
            \hline Agent & Importance Order \\\hline
            1 & $F\impord{1}B\impord{1}T$ \\
            2 & $T\impord{2}F\impord{2}B$ \\\hline
        \end{tabular}
    \end{center}
    \lps{} gives the \ram{} assignment denoted $P$.
    If agent $2$ misreports her importance order as $\impord{2}:F\impord{2}T\impord{2}B$, \lps{} gives another \ram{} assignment denoted $P'$.
    Both $P$ and $P'$ are shown as follows (To save space, we omit the columns which only contain $0$, similarly hereinafter.):


    \vspace{1em}\noindent
    \begin{minipage}{\linewidth}
	\begin{minipage}{0.3\linewidth}
    \begin{center}
        \centering
        \begin{tabular}{|c|cc|}
            \hline\multirow{2}{*}{Agent} & \multicolumn{2}{c|}{$P$}\\\cline{2-3}
             & $1_F1_B2_T$ & $2_F2_B1_T$  \\\hline
            1 & 1 & 0 \\
            2 & 0 & 1 \\\hline
        \end{tabular}
    \end{center}
	\end{minipage}
	\hspace{0.1\linewidth}
	\begin{minipage}{0.5\linewidth}
    \begin{center}
        \centering
        \begin{tabular}{|c|cccc|}
            \hline\multirow{2}{*}{Agent} & \multicolumn{4}{c|}{$P'$}\\\cline{2-5}
             & $1_F1_B2_T$ & $1_F2_B1_T$ & $2_F1_B2_T$ & $2_F2_B1_T$ \\\hline
            1 & 0.5 & 0 & 0.5 & 0 \\
            2 & 0 & 0.5 & 0 & 0.5 \\\hline
        \end{tabular}
    \end{center}
	\end{minipage}
	\end{minipage}\vspace{1em}

    We observe that compared with $P$, agent $2$ loses $0.5$ shares of $2_F2_B1_T$, but acquires $0.5$ shares of $1_F2_B1_T$ in $P'$.
    Since $1_F2_B1_T\succ_22_F2_B1_T$, we obtain that $P' \sd{2} P$, but $P \sd{2} P'$ is false, which means \lps{} does not satisfy \sdsp{} when an agent can misreport her importance order.\qed
\end{proof}

\section{\mps{} for \mtaps{} with divisible items}
\label{MPS}
In this section we only consider \mtap{} with divisible items, which is not affected by \Cref{thm:imp}.
We present a simplified version of \mps{}~\cite{Wang19:Multi} in {\bf\Cref{alg:mps}}, since we no longer need to deal with partial preferences.
At a high level, in \mps{} agents consume bundles consisting of items in contrast with PS where agents consume items directly.
Under strict linear preferences, we prove that \mps{} satisfies \lexopt{}, which implies \sdopt{}, and provide two characterizations involving \leximin{} and \iof{} respectively.

\subsection{\textbf{Algorithm for \mps{}}}

Given an \mtap{}, \mps{} proceeds in multiple rounds as follows:
In the beginning of each round, $M'$ contains all items that are unexhausted.
Agent $j$ first decides her most preferred {\em available} bundle $\topb{j}{}$ according to $\succ_j$.
A bundle $\vx$ is available so long as every item $o\in\vx$ is unexhausted.
Then, each agent consumes the bundle by consuming all of the items in it at a uniform rate of one unit per unit of time. The round ends whenever one of the bundles being consumed becomes unavailable because an item being consumed has been exhausted.
The algorithm terminates when all the items are exhausted.

\begin{algorithm}
	\begin{algorithmic}[1]
		\State {\bf Input:} An \mtap{} $(N,M)$ and a preference profile $R$.
		\State For every $o\in M$, $\s(o)\gets 1$. $M' \gets M$. $P\gets 0^{n\times|\md|}$.
		\While{$M'\neq\emptyset$}
		\State {\bf Identify top bundle} $\topb{j}{}$ for every agent $j\le n$.
		\parState{ {\bf Consume.}
			\begin{itemize}
				\item[5.1:] For any $o\in M'$, $\con{o}{}\gets|\{j\in N:o\in\topb{j}{}\}|$.
				\item[5.2:] $\prog{}\gets \min_{o\in M'}\frac{\s(o)}{\con{o}{}}$.
				\item[5.3:] For each $j\le n$, $p_{j,\topb{j}{}}\gets p_{j,\topb{j}{}}+\prog{}$.
				\item[5.4:] For each $o\in M'$, $\s(o)\gets\s(o)-{\prog{}\times \con{o}{}}$.
			\end{itemize}
		}\vspace{-1.5em}
		\State $\ex{}\gets\arg\min_{o\in M'}\frac{\s(o)}{\con{o}{}}$, $M'\gets M'\setminus B$
		\EndWhile
		\State \Return $P$
	\end{algorithmic}
	\caption{\label{alg:mps} \mps{} for \mtaps{} under strict linear preference.}
\end{algorithm}

\begin{example}\label{eg:mps}
	The execution of \mps{} for the following instance of \mtap{} is shown in {\bf\Cref{fig:mps}}.
	
	\vspace{1em}\noindent
    \begin{minipage}{\linewidth}
    \begin{minipage}{0.4\linewidth}
	\begin{center}
        \begin{tabular}{|c|c|}
            \hline Agent & Preferences \\\hline
            1 & $1_F1_B\succ_11_F2_B\succ_12_F2_B\succ_12_F1_B$ \\
            2 & $1_F2_B\succ_22_F1_B\succ_21_F1_B\succ_22_F2_B$ \\\hline
        \end{tabular}
    \end{center}
	\end{minipage}
	\hspace{0.1\linewidth}
	\begin{minipage}{0.4\linewidth}
    \begin{center}
        \centering
        \begin{tabular}{|c|cccc|}
            \hline\multirow{2}{*}{Agent} & \multicolumn{4}{c|}{$P$}\\\cline{2-5}
             & $1_F1_B$ & $1_F2_B$ & $2_F1_B$ & $2_F2_B$ \\\hline
            1 & 0.5 & 0 & 0 & 0.5 \\
            2 & 0 & 0.5 & 0.5 & 0 \\\hline
        \end{tabular}
    \end{center}
	\end{minipage}
	\end{minipage}\vspace{1em}
	
	At round $1$, agent $1$'s top bundle is $1_F1_B$, and agent $2$'s top bundle is $1_F2_B$. Notice that both agents wish to consume $1_F$. Therefore, round $1$ ends as $1_F$ gets exhausted with both agents getting a share of $0.5$ units of $1_F$. Agents $1$ and $2$ also consume $1_B$ and $2_B$ respectively at the same rate during round $1$. At the end of round $1$, agents $1$ and $2$ are assigned with a $0.5$ share of
	$1_F1_B$ and $1_F2_B$ respectively.
	
	At the start of round $2$, there is a supply of $1$ unit of $2_F$ and $0.5$ units each of $1_B$ and $2_B$. Agent $1$'s top available bundle is $2_F2_B$ since $1_F2_B$ is unavailable for the exhausted item $1_F$, and agent $2$'s top available bundle is $2_F1_B$ accordingly. The agents consume the items of each type from their top bundles at a uniform rate and at the end of the round, all items are exhausted, and agents $1$ and $2$ have consumed $0.5$ units each of $2_F2_B$ and $2_F1_B$ respectively, which results in the final assignment as shown in {\bf\Cref{fig:mps}}, which is the un\deca{} assignment $P$ in \Cref{eg:undecomposable}.
	
	Further, we show in \Cref{rm:mpsnreal} that even under lexicographic preferences, the output of \mps{} is not always \deca{}. This means that \mps{} is only applicable to \mtaps{} with divisible items. 
\end{example}

\begin{figure}[ht]
	\centering
	\includegraphics[width=0.8\textwidth]{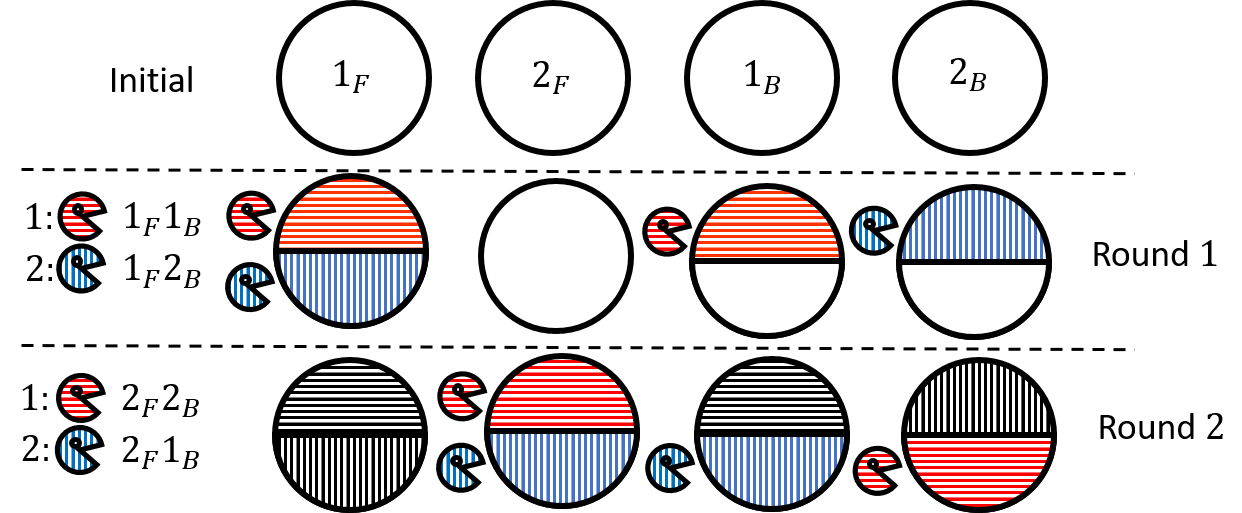}
	\caption{An example of the execution of \mps{}.}
	\label{fig:mps}
\end{figure}

\begin{remark}\label{rm:mpsnreal}
The output of \mps{} is not always a \deca{} assignment under the restriction of lexicographic preferences.
\end{remark}
\begin{proof}
For the \mtap{} in \Cref{eg:lps}, \mps{} outputs the following \ram{} assignment, denoted $P$:

\begin{center}
    \centering
    \begin{tabular}{|c|ccccccc|}
        \hline\multirow{2}{*}{Agent} & \multicolumn{7}{c|}{$P$}\\\cline{2-8}
         & $1_F1_B$ & $1_F2_B$ & $2_F1_B$ & $2_F2_B$ & $2_F3_B$ & $3_F2_B$  & $3_F3_B$\\\hline
        1 &  1/3    & 0       & 1/6       & 1/6     & 0      & 1/12    & 1/4 \\
        2 &  1/3    & 0       & 1/6       & 0       & 1/6    & 0       & 1/3 \\
        3 &  0      & 1/3     & 0         & 1/3     & 0      & 1/12    & 1/4 \\\hline
    \end{tabular}
\end{center}
When items are indivisible, if agent $2$ get $2_F3_B$, agent $1$ will get $1_F1_B$ and agent $3$ will get $3_F2_B$ as $P$ indicates.
However, $p_{1,1_F1_B}=1/3,p_{2,2_F3_B}=1/6,p_{3,3_F2_B}=1/12$ are not equal, which is a contradiction.\qed
\end{proof}

\subsection{\textbf{Efficiency and Generalized Cycles}}
Under the unrestricted domain of strict linear preferences, Theorem 2 in \cite{Wang19:Multi} implies that \mps{} satisfies \sdopt{}. We prove in~\Cref{thm:mpslexopt} below that \mps{} satisfies \lexopt{}, which is a stronger notion of efficiency than \sdopt{}, as we will prove in~\Cref{prop:gc} later.
\begin{theorem}\label{thm:mpslexopt}
    \mps{} satisfies \lexopt{}.
\end{theorem}
\begin{proof}
    Given \mtap{}(M,N) and preference profile $R$, let $P=\mps(R)$, and suppose another assignment $Q$ satisfies $Q\ld{}P$.
    Let $\hat N$ be the set of agents which have different allocations in $Q$, and by assumption we have that for any $j\in\hat N$, there exists a bundle $\vx$ satisfies $q_{j,\vx}>p_{j,\vx}$ and $q_{j,\hat\vx}=p_{j,\hat\vx}$ for any $\hat\vx\succ_j\vx$.
    W.l.o.g. let $j$ be the agent with smallest $\sum_{\hat\vx\succ_j\vx}{p_{j,\hat\vx}}$, and let $t$ denote the value.
    However, when \mps{} executes till time $t$, $\vx$ is unavailable, which means that at least one item in $\vx$ is exhausted and $p_{j,\vx}$ cannot be increased anymore.
    Therefore if agent $j$ gains more shares of $\vx$ in $Q$, there is another agent who lose shares before $t$, which is a contradiction.\qed
\end{proof}

We establish the relationship between \lexopt{} and \sdopt{} in~\Cref{prop:gc} through the {\em \ngc{}} condition (\Cref{dfn:gc}), by showing that \sdopt{} is implied by the \ngc{} condition, which is implied by \lexopt{}. We begin by borrowing the tool named generalized cycle from~\citet{Wang19:Multi}, which is based on the relation $\tau$ and the notion of {\em improvable tuples} defined as follows:
\begin{definition}{\bf{(improvable tuples~\cite{Wang19:Multi})}}
Given a fractional assignment $P$, and a profile $R=(\succ_j)_{j\le n}$,
\begin{itemize}
\item for pair of bundles $\vx,\hat{\vx}\in\md$, $\vx\tau\hat{\vx}$ $\iff$ there exists an agent $j\le n$, such that $\vx\succ_j\hat{\vx}$ and $p_{j,\hat{\vx}}>0$. \item for any pair of bundles $\vx,\hat{\vx}\in\md$, $(\vx,\hat{\vx})$ is an {\em improvable tuple} if and only if $\vx\tau\hat{\vx}$, and
\item $\gcycle(P,R)$ is the set of all improvable tuples admitted by assignment $P$ w.r.t. the preference profile $R$.
\end{itemize}
\end{definition}

For ease of exposition, we use $\gcycle(P)$ to refer to the set of all improvable tuples admitted by the fractional assignment $P$ when the profile is clear from the context. We are now ready to formally define the \ngc{} condition.
\begin{definition}{\bf{(\ngc{}~\cite{Wang19:Multi})}}\label{dfn:gc}
	Given an \mtap{} $(N,M,R)$ and a \ram{} assignment $P$, a set $C \subseteq \gcycle(P,R)$ is a {\em generalized cycle} if it holds for every $o \in M$ that: if an improvable tuple $(\vx_1,\hat\vx_1)\in C$ satisfies that $o\in \vx_1$, then there exists a tuple $(\vx_2,\hat\vx_2)\in C$ such that $o\in \hat\vx_2$. 
	An assignment $P$ satisfies {\em \ngc{}}, if it admits no generalized cycles.
\end{definition}

When $p=1$, \citet{Bogomolnaia01:New} proved that an assignment is \sdopta{} if and only if the relation $\tau$ on it is acyclic, i.e. there does not exists $\vx_1\tau \vx_2 \tau\cdots\tau\vx_1$. However, this condition fails for \mtaps{}.~\Cref{eg:gc} shows that an assignment which does not satisfy \sdopt{} satisfies the acyclicity of $\tau$, but admits a generalized cycle, which means the generalized cycle is more reliable in identifying \sdopta{} assignments.

\begin{example}\label{eg:gc}
	We illustrate generalized cycles, with the following assignment $Q$ for the \mtap{} in \Cref{eg:mps}. Note that $Q$ is not \sdopta{} because the assignment $P$ in \Cref{eg:mps} stochastically dominates $Q$.
	
    \vspace{1em}\noindent
    \begin{minipage}{\linewidth}
	\begin{minipage}{0.4\linewidth}
    \begin{center}
        \centering
        \begin{tabular}{|c|cccc|}
            \hline\multirow{2}{*}{Agent} & \multicolumn{4}{c|}{$Q$}\\\cline{2-5}
             & $1_F1_B$ & $1_F2_B$ & $2_F1_B$ & $2_F2_B$ \\\hline
            1 &  0.4    & 0   & 0     & 0.6 \\
            2 &  0.2    & 0.4   & 0.4   & 0 \\\hline
        \end{tabular}
    \end{center}
	\end{minipage}
	\hspace{0.1\linewidth}
	\begin{minipage}{0.4\linewidth}
	\centering
    \begin{center}
        \begin{tabular}{|c|c|}
            \hline Agent & Improvable Tuples \\\hline
            1 & \tabincell{c}{$(1_F1_B,2_F2_B),(1_F2_B,2_F2_B),$\\$(2_F1_B,2_F2_B)$} \\\hline
            2 & \tabincell{c}{$(1_F2_B,2_F1_B),(1_F2_B,1_F1_B),$\\$(2_F1_B,1_F1_B)$} \\\hline
        \end{tabular}
    \end{center}
	\end{minipage}
	\end{minipage}\vspace{1em}
	

	It is easy to see that $\tau$ is acyclic on $Q$. 
    However, it is easy to verify that there is a generalized cycle on $Q$: $\{(1_F1_B,2_F2_B),\allowbreak(1_F2_B,2_F1_B),\allowbreak(2_F1_B,1_F1_B)\}$.
	Consider for example the items in type $B$: $1_B$ is present in both components of $(2_F1_B,1_F1_B)$, and $2_B$ is present in the left component of $(1_F2_B,2_F1_B)$ and the right component of $(1_F1_B,2_F2_B)$. A similar analysis can be performed for the items of type $F$.
\end{example}

~\Cref{prop:gc} reveals the relationship between \lexopt{} and \sdopt{} vis-\`a-vis the \ngc{} condition. Unlike~\citet{Bogomolnaia15:Random} who point out that \lexopt{} and \sdopt{} are equivalent in their setting where $p=1$, because both of them are equivalent to acyclicity condition on the relation $\tau$, we show that this is no longer true for \mtaps{}.

\begin{proposition}\label{prop:gc}
	Given a preference profile $R$ and a \ram{} assignment $P$,
	\begin{enumerate}[label=(\arabic*)]
	    \item \label{prop:gcsdopt}$P$ is \sdopta{} regarding $R$ if $P$ admits no generalized cycle.
	    \item $P$ is \lexopta{} regarding $R$ only if $P$ admits no generalized cycle.
	\end{enumerate}
\end{proposition}

\begin{proof}
    (1) The proof is similar to the proof of Theorem 5, Claim (1) in \citet{Wang19:Multi}. A full proof is provided in \Cref{sec:prop:gc} for completeness.

    (2) Suppose by contradiction there is a \lexopta{} assignment $P$ which admits a generalized cycle $C$.
    For any agent $j$, let $(\vx,\hat\vx)\in C$ be the one of tuple she involves, and let $\vx$ be top ranked among bundles in all the tuples involved $j$.
    For every $o_k\in\vx$, we can find a tuple $(\vx_k,\hat\vx_k)\in C$ which satisfies that $o\in\hat\vx_k$ and $p_{\hat j,\hat\vx_k}>0$ for some $\hat j$ by definition.
    Then we extract shares of each $\hat\vx_k$ by a small enough value $\epsilon$, and we can have $\epsilon$ shares of $\vx$ with the shares of each $o_k$ from $\hat\vx_k$.
    And to keep the supply of bundle not beyond her demand, she should give out the same shares of $\hat\vx$.
    It is trivial that the shares of items in $\hat\vx$ and rest items in each $\hat\vx_k$ can be combined as bundles with no share left for single items, and we just assign them arbitrarily to agents who donate $\hat\vx_k$ to meet their demands.
    We note that $\epsilon$ is chosen to be small enough so that the shares of bundles above are not used up, and therefore they can be used for other agents.
    Let $Q$ be the new assignment after we do the step above for all the agents.
    In this way, for any agent $j$ with the chosen tuple $(\vx,\hat\vx)$, it gains shares of $\vx$ and maybe other bundles, with loss of shares of $\hat\vx$ and other bundles ranked behind $\vx$.
    It follows that $Q_{j,\vx}\ge P_{j,\vx}+\epsilon>P_{j,\vx}$ and for any $\vx'\succ_j\vx$, $Q_{j,\vx'}\ge P_{j,\vx'}$ because agent $j$ gains but not loses the shares of these bundles, which means $Q\ld{}P$, a contradiction.\qed
\end{proof}

\begin{remark}\label{rm:gcnne}
    The \ngc{} condition is not a necessary condition of \sdopt{}.
    Consider an \mtap{} with two agents where $\succ_1,\succ_2$ are the same as $1_F1_B\succ_11_F2_B\succ_12_F1_B\succ_12_F2_B$.
    We find that the assignment $P$ in \Cref{eg:mps} is \sdopta{} for this \mtap{} but admits a generalized cycle (not unique): $\{(1_F1_B,1_F2_B),\allowbreak(1_F2_B,2_F1_B),\allowbreak(2_F1_B,2_F2_B)\}$.
    The \ngc{} condition is also not a sufficient condition of \lexopt{}.
    With the $\succ_1,\succ_2$ above, the following assignment admits no generalized cycle, but is not \lexopta{}:
    \begin{center}
        \centering
        \begin{tabular}{|c|cccc|}
            \hline Agent & $1_F1_B$ & $1_F2_B$ & $2_F1_B$ & $2_F2_B$ \\\hline
            1,2 &  0    & 0.5   & 0.5   & 0 \\\hline
        \end{tabular}
    \end{center}
\end{remark}


In~\Cref{thm:eps}, we characterize the entire set of assigments that do not admit generalized cycles by the family of {\em eating} algorithms for \mtaps{} (Algorithm~\ref{alg:eps}), which is a natural extension of the family of eating algorithms introduced in~\citet{Bogomolnaia01:New} for the single type setting. Each eating algorithm is specified by a collection of exogenous {\em eating speed functions} $\omega=(\omega_j)_{j\le n}$. An eating speed function $\omega_j$ specifies the instantaneous rate at which agent $j$ consumes bundles, consisting of an item of each type, at each instant $t\in[0,1]$, such that the integral $\int_{t=0}^1\omega_j(t)$ is $1$. In an eating algorithm, in each round, each agent $j$ consumes her most preferred available bundle at the rate specified by her eating speed function $\omega_j$, until the supply of one of the items in one of the bundles being consumed is exhausted. Note that \mps{} is a special case of the family of eating algorithms, with $\omega_j(t)=1$ for all $t\in[0,1]$, and all $j\le n$.

\begin{algorithm}
	\begin{algorithmic}[1]
		\State {\bf Input:} An \mtap{} $(N,M)$ and a preference profile $R$.
		\State {\bf Parameters:} Eating speed functions $\omega=(\omega_j)_{j\le n}$.
		\State For every $o\in M$, $\s(o)\gets 1$. $M' \gets M$. $P\gets 0^{n\times|\md|}$. $t \gets 0$.
		\While{$M'\neq\emptyset$ and $t<1$}
		\State {\bf Identify top bundle} $\topb{j}{}$ for every agent $j\le n$.
		\parState{ {\bf Consume.}
			\begin{itemize}
				\item[5.1:] For any $o\in M'$, $\con{o}{}\gets|\{j\in N:o\in\topb{j}{}\}|$.
				\item[5.2:] $\prog{}\gets \min\{\prog{}|\sum_{j\in\con{o}{}}\int_t^{t+\prog{}}\omega_j=\s(o),o\in M'\}$.
				\item[5.3:] For each $j\le n$, $p_{j,\topb{j}{}}\gets p_{j,\topb{j}{}}+\int_t^{t+\prog{}}\omega_j$.
				\item[5.4:] For each $o\in M'$, $\s(o)\gets\s(o)-{\sum_{j\in\con{o}{}}\int_t^{t+\prog{}}\omega_j}$.
			\end{itemize}
		}\vspace{-1.5em}
		\State $M'\gets M'\setminus \{o\in M'|\s(o)=0\}$. $t\gets t+\prog{}$.
		\EndWhile
		\State \Return $P$
	\end{algorithmic}
	\caption{\label{alg:eps} Eating Algorithms}
\end{algorithm}

\begin{theorem}\label{thm:eps}
	An assignment satisfies \ngc{} if and only if the assignment is the output of an eating algorithm (Algorithm~\ref{alg:eps}).
\end{theorem}

\begin{proof}

    $(\Leftarrow)$ The proof that an assignment $P$, which is an output of an eating algorithm satisfies \ngc{} is similar to the proof of \Cref{prop:gc}, Claim~\ref{prop:gcsdopt}.


    $(\Rightarrow)$ Let $R$ be an arbitrary preference profile, and let $P$ be an arbitrary assignment satisfying \ngc{} w.r.t. $R$. For convenience, we define some quantities to represent the state during the execution of a member of the family of eating algorithms at each round $s$. For ease of exposition, we use $s=0$ to represent the initial state before the start of execution. Let $M^0 = M$, $\md^0 = \md$. Let $B^s = \{ o \in M^{s-1} : \text{there are no } \vx,\allowbreak \hat \vx \in \md^{s-1} \text{ s.t. } o \in \hat \vx \text{ and } (\vx, \hat \vx) \in \gcycle(P) \}$, $M^s = M^{s-1} - B^s$ and $\md^s = \{ \vx \in \md : \text{for every } o \in \vx, o \in M^s \}$ be the available bundles in $M^s$. Let $S = \min\{s : M^s = \emptyset\}$.
    Let $\omega_j(t)$ be the eating rate of agent $j$ at time $t$, $N(\vx,\md^s)$ be the set of agents who prefer $\vx$ best in the available bundles $\md^s$ for any $\vx \in \md^s$. The following eating speed functions $\omega_j$ define the algorithm:
    \[
    \forall s \leq S, \frac{s-1}{S} \leq t \leq \frac{s}{S}, \omega_j(t)\overset{\text{def}}{=}
    \left\{
    \begin{aligned}
     S \times p_{j,\vx}, \qquad & \exists o \in \vx, o \in B^s \text{ and } j \in N(\vx, \md^{s-1}) \\
     0,  \qquad & \text{otherwise.}
    \end{aligned}
    \right.
    \]
    From the design of algorithm, we know that items in $B^s$ decide which bundles in $\md^{s-1}$ are consumed in round $s$, and these items are not consumed after the round $s$.
    We claim that the algorithm specified by the eating speed functions $\omega=(\omega_j)_{j\le n}$ above, outputs $P$ for the \mtap{} with preference profile $R$.
    Let $Q$ be the output of the eating algorithm and we prove $P = Q$ by proving with induction that for any $s\in[1,S]$, $o \in B^s$, $\vx \in \md^{s-1}$ s.t. $o \in \vx$, we prove for each $j \in N, p_{j,\vx} = q_{j,\vx}$. The base case where $s=1$ is omitted here, and we provide a proof sketch for the inductive step in the following.
    The full proof is in \Cref{sec:thm:eps}.



    \noindent{\bf Inductive step.}
    With the assumption that $p_{j,\vx} = q_{j,\vx}$ for any $o \in \bigcup_{k = 1}^s B^k$, $\vx \in \md$ s.t. $o \in \vx$ and any $j \in N$, we prove that: for any $o \in B^{s+1}$, $\vx \in \md$ s.t. $o \in \vx$ and any $j \in N$, $p_{j,\vx} = q_{j,\vx}$. If $\vx$ is not in $\md^s$, then there is an item $\hat o \in \bigcup_{k = 1}^s B^k$ s.t. $\hat o \in \vx$ and we have $p_{j,\vx} = q_{j,\vx}$ by the assumption.
    For $o\in B^{s+1}$ and $\vx\in\md^{s}$ s.t. $o\in\vx$, if $\vx$ is not most preferred in $\md^s$ by some agent $j$, then $p_{j,\vx}=0$, because if $p_{j,\vx}>0$, then let $\hat\vx$ be the most preferred bundle by $j$ in $\md^s$ and we have $(\hat\vx,\vx)\in\gcycle(P)$ where $o\in\vx$, a contradiction to the construction of $B^{s+1}$.
    Therefore we have the following equation by construction
    \[
    \sum_{o\in\vx,\vx\in\md^s}\sum_{j\in N(\vx,\md^s)}p_{j,\vx} + \sum_{o\in\vx,\vx\in \md\setminus\md^s}\sum_{j \in N}p_{j,\vx} = \sum_{o\in\vx,\vx\in\md}\sum_{j\in N}p_{j,\vx} = 1.
    \]
    That implies for any $o\in B^{s+1}$, $\vx \in \md^s$ s.t. $o \in \vx$, if $j$ is not in $N(\vx, \md^s)$, then $p_{j,\vx} = 0$.

    For any $o \in B^{s+1}$, $\vx \in \md^s$ s.t. $o \in \vx$ and $j \in N(\vx, \md^s)$, we prove that agent $j$ consumes exactly $p_{j,\vx}$ units of bundle $\vx$. By the assumption, $o$ remains available in $\frac{s}{S} \leq t \leq \frac{s+1}{S}$ and agent $j$ consumes $\vx$ in $\frac{s}{S} \leq t \leq \frac{s+1}{S}$. We know that $q_{j,\vx}=0$ by time $\frac{s}{S}$, and by time $\frac{s+1}{S}$, $q_{j,\vx} = \frac{1}{S} \times S \times p_{j,\vx} = p_{j, \vx}$. That means
    \begin{equation*}
        \begin{split}
            \sum_{o\in\vx,\vx\in\md^s}\sum_{j\in N(\vx,\md^s)}q_{j,\vx}  + \sum_{o\in\vx,\vx\in\md\setminus\md^s}\sum_{j \in N}q_{j,\vx} =
    \sum_{o\in\vx,\vx\in\md^s}\sum_{j\in N(\vx,\md^s)}p_{j,\vx} + \sum_{o\in\vx,\vx\in\md\setminus\md^s}\sum_{j \in N}p_{j,\vx} = 1
        \end{split}
    \end{equation*}
    That implies for any $o \in B^{s+1}$, $\vx \in \md^s$ s.t. $o \in \vx$,
    \begin{enumerate*}[label=(\roman*)]
    \item if $j \in N(\vx, \md^s)$, then $q_{j,\vx} = p_{j, \vx}$;
    \item if $j$ is not in $N(\vx, \md^s)$, then $q_{j,\vx} = 0 = p_{j, \vx}$.
    \end{enumerate*}
    Therefore, for any $o \in B^{s+1}$, $\vx \in \md^s$ s.t. $o \in \vx$ and any $j \in N$, $p_{j,\vx} = q_{j,\vx}$.
    \qed
\end{proof}

\subsection{\textbf{Fairness and Characterization}}
For CP-net preferences, Theorem~5 in~\citet{Wang19:Multi} showed that \mps{} satisfies \sdef{} for CP-net preferences.
Here CP-net determines the dependence among preferences of types, which also reflects the importance of each type.
Since the domain of CP-net preferences and strict linear preferences are not totally overlapping, we provide \Cref{prop:mps} as complement and the proof is in \Cref{sec:prop:mps}.

\begin{proposition}\label{prop:mps}
	\mps{} satisfies \sdef{} for \mtaps{}.
\end{proposition}

By \cite{Wang19:Multi}, MPS is \sdspa{} if all the agents share the same CP-net, which means their importance orders are identical.
We prove that under lexicographic preferences, \mps{} satisfies \sdsp{}, and in particular, the importance orders of agents can be different.

\begin{theorem}\label{thm:mpssp}
	\mps{} satisfies \sdsp{} for \mtaps{} with lexicographic preferences.
\end{theorem}
\begin{proof}
	Consider an arbitrary \mtap{} $(N,M)$ and an arbitrary lexicographic preference profile $R$. Suppose for the sake of contradiction that an agent $j$ can obtain a better allocation by misreporting her preference as another lexicographic preference $\succ'_j$. Throughout, we use $P=\mps{}(R)$ and $Q=\mps{}(R')$, where $R'=(\succ'_j,\succ_{-j})$. By assumption of beneficial misreporting, we have $Q\sd{j}P$. We will show that $Q_j =P_j$.
	
	We claim that $\srd{Q}{i}_j=\srd{P}{i}_j$ for every type $i \leq p$, where $\srd{P}{i}_j$ is agent $j$'s single type allocations of type $i$ at $P$ and $\srd{Q}{i}_j$ the same at $Q$.
	W.l.o.g. let the types be labeled such that $1 \impord{j} \cdots \impord{j} p$.
	For convenience, for any type $i\le p$, we define $\itemwt{i}$ to be an item $o$ of type $i$, and $(\itemwt{i}*)$ to be an arbitrary bundle containing $\itemwt{i}$.

	\begin{clm}\label{cl:mpstops}
		Under lexicographic preferences, $(\mps{}(R))^i=PS(R^i)$, where ${R^i=(\prefintype{j}{i})_{j\leq n}}$.
	\end{clm}
	The claim is obtained by comparing the execution of \mps{} with PS in each type.
	The full proof of the claim is in \Cref{sec:cl:mpstops}.
	Since PS satisfies \sdsp{}~\cite{Bogomolnaia01:New}, we deduce from {\bf\Cref{cl:mpstops}} that for each type $i \leq p$, $\srd{Q}{i}\sd{j}\srd{P}{i} \Rightarrow \srd{Q}{i}_j=\srd{P}{i}_j$.
	Therefore we can prove $\srd{Q}{i}_j=\srd{P}{i}_j$ by showing $\srd{Q}{i}\sd{j}\srd{P}{i}$ instead, by induction.

	

	First we assume $\srd{Q}{1} \nsd{j} \srd{P}{1}$.
	That means there exists $\hitemwt{1}$ such that $\sum_{o \in \ucs(\prefintype{j}{i},\hitemwt{1})}\srd{p}{1}_{j,o} > \sum_{o \in \ucs({\prefintype{j}{i},\hitemwt{1}})}{\srd{q}{1}_{j,o}}$.
    Let $\hat\vx$ be the least preferred bundle containing $\hitemwt{1}$.
    It follows that $$\sum_{o \in \ucs(\prefintype{j}{i},\hitemwt{1})}\srd{p}{1}_{j,o}=
    \sum_{\vx\in\ucs(\succ_j,\hat\vx)}{p_{j,\vx}}>
	\sum_{\vx\in\ucs(\succ_j,\hat\vx)}{q_{j,\vx}}=
	\sum_{o \in \ucs({\prefintype{j}{i},\hitemwt{1}})}{\srd{q}{1}_{j,o}},$$
    which is a contradiction to our assumption. Thus $\srd{Q}{1}_j=\srd{P}{1}_j$.
	
	Next we prove that for any type $i \leq p$ when $\srd{Q}{\hat i}_j = \srd{P}{\hat i}_j$ for every $\hat i < i$.
	For arbitrary $\hitemwt{1},\hitemwt{2},\dots,\hitemwt{i}$, let $\hat\vx=(\hitemwt{1},\dots,\hitemwt{i}*)$ be the least preferred bundle containing them.
    Because $Q\sd{j}P$, we have $\sum_{{\vx}\in \ucs(\succ_j,\hat\vx)}{p_{j,{\vx}}} \leq
    \sum_{{\vx}\in \ucs(\succ_j,\hat\vx)}{q_{j,{\vx}}}$, and we take apart the sum like:

	\begin{equation}\label{sdspi0}
	\begin{split}
    &\sum_{{\vx}\in\{(\itemwt{1}*)|\itemwt{1} \prefintype{j}{1} \hitemwt{1}\}}{p_{j,{\vx}}} +
    \sum_{{\vx}\in\{(\hitemwt{1},\itemwt{2}*)|\itemwt{2} \prefintype{j}{2} \hitemwt{2}\}}{p_{j,{\vx}}}\\
   &+\dots+\sum_{{\vx}\in\{(\hitemwt{1},\hitemwt{2},\dots,\itemwt{i-1}*)|\itemwt{i-1} \prefintype{j}{i-1} \hitemwt{i-1}\}}{p_{j,{\vx}}}+\sum_{{\vx}\in S_i}{p_{j,{\vx}}}\\
	\end{split}
	\end{equation}
   Here we use $S_i$ to refer to the set $\{(\hitemwt{1},\dots,\hitemwt{i-1},\itemwt{i}*)|\itemwt{i} \in \ucs(\prefintype{j}{i},\hitemwt{i})\}$.
   {\bf\Cref{cl:eqwtype}} follows from the observation that agents consume bundles until unavailable.
   The full proof of the claim is in \Cref{sec:cl:eqwtype}.

    \begin{clm}\label{cl:eqwtype}
        For $P=\mps(R)$, given a fixed $i'$ and $\srd{Q}{\hat i}_j=\srd{P}{\hat i}_j$ for any ${\hat i} \le i'$, if {$Q\sd{j}P$}, then for $S_{\hat i}=\{(\hitemwt{1},\dots,\hitemwt{{\hat i}-1},\itemwt{\hat i}*)|\itemwt{\hat i}\in\ucs(\prefintype{j}{\hat{i}},\hitemwt{\hat i})\},
        \sum_{\vx\in S_{\hat i}}{p_{j,{\vx}}}=\sum_{\vx\in S_{\hat i}}{q_{j,{\vx}}}$.
	\end{clm}
	
    With (\ref{sdspi0}) and {\bf\Cref{cl:eqwtype}} we obtain $\sum_{{\vx}\in S_i}{p_{j,{\vx}}} \leq
    \sum_{{\vx}\in S_i}{q_{j,{\vx}}}$. 
	We can sum up the inequality by $\hitemwt{1},\hitemwt{2},\dots,\hitemwt{i-1}$ like
	\begin{equation*}
	    \sum_{\hitemwt{1}} \sum_{\hitemwt{2}} \dots \sum_{\hitemwt{i-1}} \sum_{{\vx}\in S_i}{p_{j,\vx}} \leq \sum_{\hitemwt{1}} \sum_{\hitemwt{2}} \dots \sum_{\hitemwt{i-1}} \sum_{{\vx}\in S_i}{p_{j,{\vx}}}.
	\end{equation*}
	Then we have
	\begin{displaymath}
	\sum_{{\vx}\in\{(\itemwt{i}*)|\itemwt{i} \in \ucs(\prefintype{j}{i},\hitemwt{i})\}}{p_{j,{\vx}}} \leq \sum_{{\vx}\in\{(\itemwt{i}*)|\itemwt{i} \in \ucs(\prefintype{j}{i},\hitemwt{i})\}}{q_{j,{\vx}}}
	\end{displaymath}
	It is equal to $\sum_{\itemwt{i} \in \ucs(\prefintype{j}{i},\hitemwt{i})}{\srd{p}{i}_{j,\itemwt{i}}} \leq \sum_{\itemwt{i} \in \ucs(\prefintype{j}{i},\hitemwt{i})}{\srd{q}{i}_{j,\itemwt{i}}}$, which means $\srd{Q}{i} \sd{j} \srd{P}{i}$.
	
	We have proved that $\srd{Q}{i}_j=\srd{P}{i}_j$ for every $i \leq p$.
	By {\bf\Cref{cl:eqwtype}}, agent $j$ has the same share over each upper contour set and therefore $Q_j=P_j$.
    \qed
\end{proof}

In \Cref{thm:char} we provide two characterizations for \mps{}.
The \leximin{} reflects the egalitarian nature of the PS mechanism, which means that the mechanism always tries to balance the shares over top ranked items/bundles among all the agents.
\mps{} retains this nature for \mtaps{} with strict linear preferences.
The \iof{} is extended from \cite{Hashimoto14:Two}.
We note that \iof{} involves the cumulative shares regarding items, different from the version in \cite{Wang19:Multi} which only involves the share of each bundle.

\begin{theorem}\label{thm:char}
    \mps{} is the unique mechanism which \begin{enumerate*}[label=(\Roman*)] \item satisfies \leximin{}, and \item satisfies \iof{}.\end{enumerate*}
\end{theorem}
\begin{proof}
    \noindent{\bf (\uppercase\expandafter{\romannumeral1}) \leximin{}:} Let $R$ be any profile of strict linear preferences.
    Let $P=\mps{}(R)$, $\vu=(u_{j,\vx})_{j\le n, \vx\in\md}$.
    For $j\le n$, and bundle $\vx\in\md$, let $u_{j,\vx}=\sum_{\hat\vx\in\ucs(\succ_j,\vx)}{p_{j,\hat\vx}}$. Suppose for the sake of contradiction that $P$ is not \leximina{}. Then, there exists a \ram{} assignment $Q$, such that for $\vv=(v_{j,\vx})_{j\le n, \vx\in\md}$, where $v_{j,\vx}=\sum_{\hat\vx\in\ucs(\succ_j,\vx)}{q_{j,\hat\vx}}$, it holds that $(\vv,\vu)\in L$, where $L$ is the leximin relation (\Cref{dfn:leximin}).
    Let $\vu^*$ constructed by sorting the components of $\vu$ in ascending order, and $\vv^*$ be defined on $\vv$ similarly.

    We prove by induction that $\vv^*=\vu^*$, i.e. $v^*_k=u^*_k$ for all $k$, and the corresponding assignments $P$ and $Q$ are identical.

    \noindent{\bf Base case.} As the basis of induction, we prove that $u^*_1=v^*_1$.
    Suppose for the sake of contradiction that $u^*_1<v^*_1$.
    We use the tuple $(j,\vx)$ as the index of the component $u_{j,\vx}$, for every agent $j\le n$, and bundle $\vx\in\md$. Let $S_1$ be the set of indices, such that for each $(j,\vx)\in S_1$, $u_{j,\vx} = u^*_1$.

    We consider the corresponding elements in $\vv$ indicated by the set $S_1$.
    We note that for each $(j,\vx)\in S_1$, there are two possible cases:
    \begin{enumerate}[label={\em Case \arabic*}.,wide,labelindent=0pt,topsep=0pt]%
        \item\label{lexicase1} $\vx$ is the most preferred bundle w.r.t. $\succ_{j}$. Then, $p_{j,\vx}=\allowbreak\sum_{\vx'\in\ucs({j},{\vx})}{p_{j,\vx'}}\allowbreak=\allowbreak u^*_1$, and $u^*_1<v^*_1$, implies that $p_{j,\vx}=u^*_1<v^*_1\leq\sum_{\vx'\in\ucs({j},{\vx})}{q_{j,\vx'}}=q_{j,\vx}$, 
        \item\label{lexicase2} $\vx$ is {\em not} the most preferred bundle w.r.t. $\succ_{j}$. Then, for the most preferred bundle $\hat\vx$ w.r.t. $\succ_{j}$, there must exist $p_{j,\hat\vx}=u^*_1$ as in {\em Case 1}, since $u^*_1\allowbreak\le\sum_{\vx'\in\ucs({j},{\hat\vx})}{p_{j,\vx'}}\allowbreak\le\sum_{\vx'\in\ucs({j},{\vx})}{p_{j,\vx'}}\allowbreak=u^*_1$. This implies that $p_{j,\vx}=0\leq q_{j,\vx}$.
    \end{enumerate}

    From the execution of \mps{}, $\vx$ must be unavailable at time $u^*_1$ because some items in it are exhausted at that time.
    Let $B_1$ denote the set of the items exhausted at time $u^*_1$.
    For any $o\in B_1$, we have that $\sum_{(j',\vx')\in S_1,o\in\vx'}{p_{j',\vx'}}=1$.
    With the inequality in {\em Case 1} and {\em Case 2} we have
    $\sum_{(j',\vx')\in S_1,o\in\vx'}{q_{j',\vx'}}>\sum_{(j',\vx)\in S_1,o\in\vx'}{p_{j',\vx'}}=1$ for some $o\in B_1$, which is a contradiction.

    Having shown that $u^*_{1}=v^*_{1}$, we claim that the shares indicated by $S_1$ are equal in $P$ and $Q$, i.e. for any $(j,\vx)\in S_{1}$, $p_{j,\vx}=q_{j,\vx}$.
    Suppose for the sake of contradiction, there exists a tuple $(j,\vx)\in S_{1}$ such that $p_{j,\vx}<q_{j,\vx}$, then for any $o\in\vx$ such that $o\in B_1$, there must exist $(\hat j,\hat\vx)\in S_{1}$ such that $\hat\vx$ contains $o$ and $p_{\hat j,\hat\vx}>q_{\hat j,\hat\vx}\geq 0$, because for any $o\in B_1$ $\sum_{(j',\vx')\in S_1,o\in\vx'}{p_{j',\vx'}}=1$.
    By {\em Case 1} we know that $\hat\vx$ is most preferred by agent $\hat j$, and therefore $u^*_1=\sum_{\vx'\in\ucs(\succ_{\hat j},\hat\vx)}p_{\hat  j,\vx'}>\sum_{\vx'\in\ucs(\succ_{\hat j},\hat\vx)}q_{{\hat j},\vx'}$ which means $(\vu,\vv)\in L$, a contradiction to the assumption.
    The claim also means that for all $k\le |S_1|$, $u^*_k = u^*_1 = v^*_k$.

    \noindent{\bf Inductive step.} Next, we prove by induction that for any $k>1$, $u^*_k=v^*_k$ given $u^*_k>u^*_{k-1}$, and $u^*_{l}=v^*_{l}$ for $l<k$, and $p_{j,\vx}=q_{j,\vx}$ for any $(j,\vx)\in S_{l},l<k$.
    Suppose for the sake of contradiction that $u^*_k<v^*_k$.
    Let $S_k$ be the set of indices, such that for each $(j,\vx)$ in $S_k$, $u_{j,\vx} = u^*_k$.
    Let $S_k$ be the set of tuples which correspond to $u^*_k$.
    For $(j,\vx)\in S_k$, let $\hat\vx$ be the least preferred bundle in $\{\vx'|\vx'\succ_{j}\vx\}$ w.r.t. $\succ_j$, and its corresponding index is $(j,\hat\vx)$.
    Then, we have $p_{j,\vx}=u^*_k-u_{j,\hat\vx}$.
    Let $u^*_{l}=u_{j,\hat\vx}$.
    By our intial assumption that $(\vv,\vu)\in L$, we claim that for all $(j,\vx)\in S_k$, $p_{j,\vx}\leq q_{j,\vx}$ and strict for some indices, because
    \begin{enumerate}[label={\em Case \arabic*'}.,wide,labelindent=0pt,topsep=0pt]
        \item if $(j,\hat\vx)\notin S_k$, then we have that $l<k$ and $u^*_k>u^*_{l}=v^*_{l}$, and therefore $p_{j,\vx}=u^*_k-u^*_{l}< v^*_k-v^*_{l}=q_{j,\vx}$,
        \item if $(j,\hat\vx)\in S_k$, then $u^*_{l}=u^*_{k}$ and $p_{j,\vx}=0\leq q_{j,\vx}$.
    \end{enumerate}
    W.l.o.g let $\vx$ satisfy $p_{j,\vx}< q_{j,\vx}$.
    We know $\vx$ is unavailable at time $u^*_k$ in the execution of \mps{} because of exhausted items in it.
    Let $B_k$ be the set of items exhausted at time $u^*_k$.
    For any $o\in B_k$, we have
    \begin{equation}\label{char:3}
    \begin{split}
        \sum_{(j',\vx')\in \cup_{l<k}S_{l},o\in\vx'}{p_{j',\vx'}}&=\sum_{(j',\vx')\in \cup_{l<k}S_{l},o\in\vx'}{q_{j',\vx'}}\\
        \sum_{(j',\vx')\in \cup_{l<k}{S_{l}},o\in\vx'}{p_{j',\vx'}}&+
        \sum_{(j',\vx')\in S_k,o\in\vx'}{p_{j',\vx'}}=1.
    \end{split}
    \end{equation}
    Thus we have $\sum_{(j',\vx')\in \cup_{l\le k}S_{l},o\in\vx}{q_{j',\vx'}}>\allowbreak\sum_{(j',\vx')\in \cup_{l\le k}S_{l},o\in\vx}{p_{j',\vx'}}=1$ for some $o\in B_k$, which is a contradiction.

    Then we show if the shares indicated by $S_{k}$ in $P$ and $Q$ are equal.
    With $u^*_{k}=v^*_{k}$, we claim that
    \begin{equation}
        \text{for any }(j,\vx)\in S_{k}, p_{j,\vx}=q_{j,\vx}.\label{char:4}
    \end{equation}
    If a tuple $(j,\vx)\in S_k$ satisfies $p_{j,\vx}<q_{j,\vx}$, there must exist some $(\hat j,\hat\vx)\in S_k$ such that $p_{\hat j,\hat\vx}>q_{\hat j,\hat\vx}\ge 0$ by (\ref{char:3}) since $p_{j',\vx'}=q_{j',\vx'}$ for $(j',\vx')\in S_{l},l<k$.
    Because $p_{\hat j,\hat\vx}>0$, by {\em Case 1'} we know that $u^*_k=u_{\hat j,\hat\vx}>v_{\hat j,\hat\vx}$ which means $(\vu,\vv)\in L$, a contradiction to the assumption.
    By the claim we also know that for all $k\le l< k+|S_k|$, $u^*_l = u^*_k$.

    By induction, we have $u^*_k=v^*_k$ for all $k\leq n|\md|$ and therefore $\vu=\vv$.
    Besides, since $\bigcup_{k<=|\vu|}{S_k}=N\times\md$, from (\ref{char:4}) we have that $p_{j,\vx}=q_{j,\vx}$ for any $j\in N,\vx \in \md$.

\vspace{1em}
    \noindent{\bf (\uppercase\expandafter{\romannumeral2}) \iof{}:} We use the relation of time and consumption to make the proof and show the uniqueness.
    Given any \mtap{}(M,N) and preference profile $R$, let $P=\mps(R)$ in the following proof.

    \noindent{\bf Satisfaction:}
    For an arbitrary agent $j$ and any bundle $\vx$, let $t=\sum_{\vx'\in\ucs(\succ_j,\vx)}p_{j,\vx}$.
    We assume by contradiction that there exists an agent $k$ who gets shares of bundle $\hat\vx$ containing $o$ i.e. $p_{k,\hat\vx}>0$, and $\hat t=\sum_{\vx'\in\ucs(\succ_k,\hat\vx)}p_{k,\vx}>t$.
    From the relation of time and consumption, we know that at time $t$, $\vx$ is unavailable and therefore the supply of some item $o\in\vx$ is exhausted.
    We note that $t$ is not necessary to be the exact time when $\vx$ become unavailable.
    We also know that at time $\hat t$, $\hat\vx$ become unavailable, and therefore for $t<t'<\hat t$, $\hat\vx$ is still available, which also means $o\in\hat\vx$ is not exhausted, which is a contradiction.

    \noindent{\bf Uniqueness:}
    Suppose $Q$ is another \iofa{} assignment by contradiction.
    For some agent $j$, we can find by comparing $P$ and $Q$ that a bundle $\vx$ which satisfies $p_{j,\vx}\neq q_{j,\vx}$ and $p_{j,\hat\vx}=q_{j,\hat\vx}$ for $\hat\vx\succ_j\vx$.
    W.l.o.g. let agent $j$ and $\vx$ just mentioned satisfy that $t=min(\sum_{\vx'\in\ucs(\succ_j,\vx)}{p_{j,\vx'}},\sum_{\vx'\in\ucs(\succ_j,\vx)}{q_{j,\vx'}})$ is the smallest among all the agents and bundles.
    By the selection of $j$ and $\vx$, we have $p_{k,\hat\vx}=q_{k,\hat\vx}$ for $(k,\hat\vx)\in S=\{(k,\hat\vx)|\sum_{\vx'\in(\succ_k,\hat\vx)}{q_{j,\vx'}}\le t\}$.

    We first claim that $p_{j,\vx}> q_{j,\vx}$, i.e. $t=\sum_{\vx'\in\ucs(\succ_j,\vx)}{q_{j,\vx'}}$.
    Otherwise, if $p_{j,\vx}< q_{j,\vx}$, then agent $j$ get more shares of $\vx$ in $Q$ and therefore demands more supply of item contained in $\vx$.
    Since the assumption also means $t=\sum_{\vx'\in\ucs(\succ_j,\vx)}{p_{j,\vx'}}$, we know that in $P$ there exists an item $o\in\vx$ exhausted at $t$ which makes $\vx$ unavailable, and the supply of $o$ is also used up in $Q$ for agents and bundles in $S$.
    Therefore the extra demand of $o$ for agent $j$ on bundle $\vx$ comes from agent $k$ on bundle $\hat\vx$ containing $o$ such that $(k,\hat\vx)\in S$, which means $q_{k,\hat\vx}<p_{k,\hat\vx}$, a contradiction.

    Since $p_{j,\vx}> q_{j,\vx}$ i.e. $t=\sum_{\vx'\in\ucs(\succ_j,\vx)}{q_{j,\vx'}}$ now, which means agent $j$ gives up some shares of $\vx$ in $Q$, and thus for every $o\in\vx$ we can find out some agent $k$ and $\hat\vx$ containing $o$ such that $k$ gains more shares of $\hat\vx$ in $Q$, i.e. $q_{k,\hat\vx}>p_{k,\hat\vx}\ge 0$.
    By the definition of $t$, we have $\sum_{\vx'\in(\succ_k,\hat\vx)}{q_{j,\vx'}}> t=\sum_{\vx'\in(\succ_j,\vx)}{q_{j,\vx'}}$.
    Therefore every $o\in\vx$ does not meet the condition of \iof{}, which is contradictory to the assumption that $Q\neq P$ is \iofa{}.
    \qed
\end{proof}

\begin{remark}\label{rm:mpsnsp}
    \mps{} does not satisfy \sdsp{}.
\end{remark}
\begin{proof}
    Consider an \mtap{} with two agents where $\succ_1,\succ_2$ are as below:
    \begin{center}
        \begin{tabular}{|c|c|}
            \hline Agent & Preferences \\\hline
            1 & $1_F2_B\succ_11_F1_B\succ_12_F1_B\succ_12_F2_B$ \\
            2 & $1_F1_B\succ_22_F1_B\succ_22_F2_B\succ_21_F2_B$ \\\hline
        \end{tabular}
    \end{center}
    \mps{} outputs $P$ for this preference profile.
    If agent $1$ misreports $\succ_1$ as $\succ_1':2_F1_B\succ_11_F1_B\succ_11_F2_B\succ_12_F2_B$, then \mps{} outputs $P'$.
    Both $P$ and $P'$ are shown as below:

    \vspace{1em}\noindent
    \begin{minipage}{\linewidth}
	\begin{minipage}{0.4\linewidth}
    \begin{center}
        \centering
        \begin{tabular}{|c|cccc|}
            \hline\multirow{2}{*}{Agent} & \multicolumn{4}{c|}{$P$}\\\cline{2-5}
             & $1_F1_B$ & $1_F2_B$ & $2_F1_B$ & $2_F2_B$ \\\hline
            1 &  0    & 0.5   & 0.25     & 0.25 \\
            2 &  0.5    & 0   & 0.25     & 0.25 \\\hline
        \end{tabular}
    \end{center}
	\end{minipage}
	\hspace{0.1\linewidth}
	\begin{minipage}{0.4\linewidth}
    \begin{center}
        \centering
        \begin{tabular}{|c|cccc|}
            \hline\multirow{2}{*}{Agent} & \multicolumn{4}{c|}{$P'$}\\\cline{2-5}
             & $1_F1_B$ & $1_F2_B$ & $2_F1_B$ & $2_F2_B$ \\\hline
            1 &  0    & 0.5   & 0.5     & 0 \\
            2 &  0.5    & 0   & 0    & 0.5 \\\hline
        \end{tabular}
    \end{center}
	\end{minipage}
	\end{minipage}\vspace{1em}


   We have that $P' \sd{1} P$ and $P' \neq P$ by shifting agent $1$'s share from $2_F2_B$ to $2_F1_B$, which violates the requirement of \sdsp{} i.e. $P'_1 = P_1$ if $P' \sd{1} P$.\qed
\end{proof}

\section{Conclusion and Future Work}
\label{Conclusion}
In the paper, we first point out that it is impossible to design \sdopta{} and \sdefa{} mechanisms for \mtaps{} with indivisible items.
However, under the natural assumption that agents' preferences are lexicographic, we propose \lps{} as a mechanism which can deal with indivisible items while satisfying the desired efficiency and fairness properties of \sdopt{} and \sdef{}.

For divisible items, we show that \mps{} satisfies the stronger efficiency notion of \lexopt{} in addition to \sdef{} under the unrestricted domain of linear preferences, and is \sdspa{} under lexicographic preferences, which complement the results in~\cite{Wang19:Multi}.
We also provide two separate characterizations for \mps{} by \leximin{}{} and \iof{}. We also show that every assignment that satisfies \ngc{}, a sufficient condition for \sdopt{}, can be computed by an eating algorithm.

Characterizing the domain of preferences under which it is possible to design mechanisms for \mtaps{} with indivisible items that are simultaneously fair, efficient, and strategyproof is an exciting topic for future research. 
For divisible items, characterizations of mechanisms for \mtaps{} satisfying \sdopt{} and \sdef{}, in addition to other combinations of desirable properties, is an interesting topic for future work.
In addition, developing efficient and fair mechanisms for natural extensions of the \mtap{} problem such as settings where there are demands for multiple units of each type, or initial endowments is also an exciting new avenue for future research.




\appendix
\section{Appendix}
\subsection{\textbf{Proof of \Cref{rm:tightimp}}}
\label{sec:rm:tightimp}
    We use the \mtap{} with the following LP-tree preference profile $R$ as a tool to prove the tighter result.
    We also show $\succ_1$ by a LP-tree.

    \vspace{1em}\noindent
    \begin{minipage}{\linewidth}
	\begin{minipage}{0.4\linewidth}
    \begin{center}
        \begin{tabular}{|c|c|}
            \hline Agent & Preferences \\\hline
            1 & $1_F1_B\succ_11_F2_B\succ_12_F2_B\succ_12_F1_B$ \\
            2 & $1_F2_B\succ_21_F1_B\succ_22_F1_B\succ_22_F2_B$ \\\hline
        \end{tabular}
    \end{center}
	\end{minipage}
	\hspace{0.1\linewidth}
	\begin{minipage}{0.4\linewidth}
        \includegraphics[width=\linewidth]{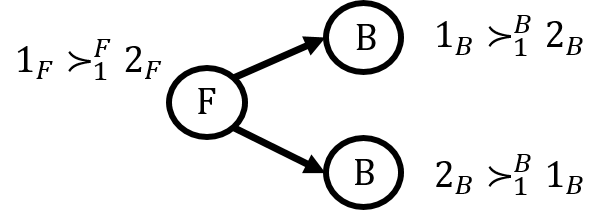}
	\end{minipage}
	\end{minipage}\vspace{1em}

    Suppose that $f$ satisfies \sdwopt{} and \sdwef{}, and let $Q=f(R)$.
    As in the proof of \Cref{thm:imp}, we show that such an \deca{} assignment $Q$ in the following does not exist.
    \begin{center}
        \centering
        \begin{tabular}{|c|cccc|}
            \hline\multirow{2}{*}{Agent} & \multicolumn{4}{c|}{$Q$}\\\cline{2-5}
             & $1_F1_B$ & $1_F2_B$ & $2_F1_B$ & $2_F2_B$ \\\hline
            1 &  $v$    & $w$   & $y$     & $z$ \\
            2 &  $z$    & $y$   & $w$     & $v$ \\\hline
        \end{tabular}
    \end{center}

    From agent $1$'s preference, we observe that she does not get shares of $1_F2_B$ and $2_F1_B$ simultaneously, because that means she can mix up the items in them to get $1_F1_B$ that better than $1_F2_B$ and $2_F2_B$ that better than $2_F1_B$.
    With the similar reason, agent $2$ does not get shares of $1_F1_B$ and $2_F2_B$ simultaneously.
    This means that in the assignment $Q$, there are at least two variants are $0$: either $v$ or $z$ is $0$, and either $y$ or $w$ is $0$.

    We first exclude the case that $v,w,y,z$ are all $0$ for $v+w+y+z=1$, and consider the cases that there are three variants are $0$, which means $Q$ is a discrete assignment. The assignments that assign any agent with her least preferable bundle, i.e. $2_F1_B$ to agent $1$ or $2_F2_B$ to agent $2$, are excluded because it violates \sdwef{}.
    The other possible cases are:
    If $Q$ assigns $1_F2_B$ to agent $1$ and $2_F1_B$ to agent $2$, then agent $2$ envies agent $1$ due to $1_F2_B\succ_22_F1_B$;
    If $Q$ assigns $2_F2_B$ to agent $1$ and $1_F1_B$ to agent $2$, then agent $1$ envies agent $2$ due to $1_F1_B\succ_12_F2_B$.
    Both cases violate \sdwef{}.

    Therefore the only possibility is that two variants are $0$, and we list all the possible combination of variants as following and explain in brief why they fail to meet the requirements:

    ($v\neq0, y\neq0$): There exists a generalized cycle ${(2_F2_B,2_F1_B),(2_F1_B,2_F2_B)}$;

    ($v\neq0, w\neq0$): $Q_1\neq Q_2$ and $Q_1\sd{2} Q_2$;

    ($y\neq0, z\neq0$): $Q_2\neq Q_1$ and $Q_2\sd{1} Q_1$;

    ($w\neq0, z\neq0$): There exists a generalized cycle ${(1_F1_B,1_F2_B),(1_F2_B,1_F1_B)}$.

    By (1) in \Cref{prop:gc}, the existence of generalized cycle means violating \sdopt{}, i.e. \sdwopt{} in this problem since there are only two agents.
    Therefore we can conclude that such a mechanism $f$ does not exist.\qed

\subsection{\textbf{Proof of (1) in \Cref{prop:gc}}}
\label{sec:prop:gc}
\begin{proof}
	The proof involves showing that every \ram{} assignment which is not \sdopta{} admits a generalized cycle. Let $P$ be such as \ram{} assignment for a given \mtap{}. Then, there exists another \ram{} assignment $Q\neq P$ such that $Q \sd{\null} P$.
	We show that the set of tuples which shows the difference between $P$ and $Q$ is a generalized cycle on $P$.
	
	Let $\hat N=\{j \leq n:P_j\neq Q_j\}\subseteq N$. By definition of \sdopt, for every $j\in \hat N$, $Q\sd{j}P$.
	Let $C$ be the set of all tuples $(\vx,\hat \vx)$ such that $\vx\succ_j \hat \vx$ and $q_{j,{\vx}}>p_{j,{\vx}}, q_{j,\hat \vx}<p_{j,\hat \vx}$ for some $j \in \hat N$. Therefore at a high level, from $C$ we can learn all the difference of shares for bundles in allocations of agents in $P$ and $Q$. First we prove the following claim:
	\begin{clm}\label{cl:gc}
		For any agent $j \in \hat N$, there exists a bundle with a greater share in $Q_j$ than in $P_j$, and this bundle is ranked above all bundles with a smaller share.
	\end{clm}
		In fact, we can prove the claim by supposing for the sake of contradiction that there is a bundle $\hat\vx$ with a smaller share in $Q$ which is ranked above all bundles with a greater share, i.e. for every ${\vx} \succ_j {\hat\vx}$, $p_{j,\hat\vx}> q_{j,\hat\vx}$ and $p_{j,\vx}\geq q_{j,\vx}$. This implies that $\sum_{\vx \in \ucs(\succ_j,\hat\vx)}p_{j,\vx}>\sum_{\vx \in \ucs(\succ_j,\hat\vx)}q_{j,\vx}$, which is a contradiction to our assumption that $Q \sd{\null} P$.
	
    Then We show that $C$ is not empty.
    If $q_{j,\vx}\geq p_{j,\vx},j\in\hat N$ for every $\vx\in\md$, then $P_j=Q_j$, a contradiction to our initial assumption.
    Thus there exists $\hat\vx$ which satisfies $q_{j,\hat\vx}<p_{j,\hat\vx}$, which means there is a tuple $(\vx,\hat\vx) \in C$ according to {\bf\Cref{cl:gc}}.
	Therefore, $C\neq \emptyset$ since $P\neq Q$.
	Since $p_{j,\hat\vx}>q_{j,\hat\vx}\ge 0$, we have that $(\vx,\hat\vx)\in\gcycle(P)$, which implies that $C\subseteq\gcycle(P)$.
	
	Suppose for sake of contradiction that $C$ is not a generalized cycle.
	W.l.o.g. let $o$ be an item such that for any bundle $\vx$ containing $o$, $\vx$ is always the left component of any improvable tuple involving $\vx$ in $C$.
    If so, we have that:
    \begin{enumerate}[label=(\roman*),itemindent=2em,topsep=0em]
        \item for $j \notin \hat N$, $q_{j,\vx}= p_{j,\vx}$ trivially by the assumption,
        \item for $j \in \hat N$, $q_{j,\vx}\geq p_{j,\vx}$.
        Because if $q_{j,\vx_0}< p_{j,\vx_0}$  for some $\vx_0$ containing $o$, then by {\bf\Cref{cl:gc}} there will be a tuple $({\hat\vx_0},{\vx_0}) \in \gcycle(P)$ where $\vx_0$ is the right component and $\vx_0$ contains $o$, which is a contradiction.
        Specifically, for $(\vx,\hat\vx)\in C$ where $o \in \vx$, $q_{j,\vx}>p_{j,\vx}$ exists for some $j\in\hat N$ because otherwise $q_{j,\vx}=p_{j,\vx}$ and therefore $\vx$ is not in any tuple in $C$.
    \end{enumerate}
	Therefore we have $\sum_{j \leq n,o \in \vx}q_{j,{\vx}}>\sum_{j \leq n,o \in \vx}p_{j,\vx}=1$, a contradiction to our assumption that $Q$ is a \ram{} assignment.
	That means $C$ satisfies one of the requirements of generalized cycle: If an tuple $(\vx_1,\hat\vx_1)\in C$ satisfies that $o\in \vx_1$, then there exists a tuple $(\vx_2,\hat\vx_2)\in C$ such that $o\in \hat\vx_2$.
	Thus $C$ is a generalized cycle on $P$.\qed
\end{proof}

\subsection{\textbf{Proof of \Cref{prop:mps}}}
\label{sec:prop:mps}
\begin{proof}
	We begin our proof with any \mtap{}$(N,M)$ with an arbitrary preference $R$. Throughout, we use $P$ to refer to \mps{}(R).
    We first give a few observations about \mps{} which are helpful for understanding the following proof.
    We know that \mps{} executes multiple rounds which end when some items are exhausted and we label all the rounds by time.
    Let $j$ be an arbitrary agent.
    Here we consider the round $r_{k_m} \in \{r_{k_1},r_{k_2},\dots\}$ where agent $j$ stop consuming a bundle.
    We note that these rounds are not necessarily continuous because agent $j$ may not change her current most preferred bundles for every round.
    Without loss of generality, let $r_{k_m}<r_{k_{\hat m}}$ if $m<\hat m$.

    Let $\vx_{k_m}$ denote the bundle consumed by $j$ at round $r_{k_m}$ of \mps{}. Let $t_{k_0}=0$, and for any round $r_{k_m}$, let $t_{k_m}$ be the units of time elapsed from the start of the mechanism till the end of round $r_{k_m}$. Then, by construction of \mps{}, we have that for any round $r_{k_m},m>1$, $p_{j,\vx_{k_m}}=t_{k_m}-t_{k_{m-1}}$.
    Specially, when $m=1$, $p_{j,\vx_{k_1}}=t_{k_1}-t_{k_0}$ trivially. This implies that

    \begin{equation}\label{cm}
    t_{k_m}=t_{k_m}-t_{k_0}=\sum_{\hat m=1}^{m}(t_{k_{\hat m}}-t_{k_{\hat m-1}}) = \sum_{\vx \in \ucs(\succ_j,\vx_{k_m})} p_{j,\vx}.
    \end{equation}

    For any round $r_{k_m}$, and any $\hat \vx$ such that $\vx_{k_m} \succ_j \hat \vx \succ_j \vx_{k_{m+1}}$, $\hat \vx$ is not consumed by $j$ i.e. $p_{j,\hat\vx}=0$. Therefore it must hold that $\hat \vx$ is unavailable by the end of round $r_{k_m}$. Let $\hat t$ denote the time at which $\hat \vx$ becomes unavailable. Then,

    \begin{equation}\label{cm2}
    \hat t \leq t_{k_m} =\sum_{\vx \in \ucs(\succ_j,\vx_{k_m})}p_{j,\vx}= \sum_{\vx \in \ucs(\succ_j,\hat \vx)}p_{j,\vx}.
    \end{equation}
	Suppose for the sake of contradiction that there is a pair of agents $j,\hat j$ such that $j$ envies $\hat j$ at $P$. Then, $P_j\nsd{j}P_{\hat j}$ and there exists a bundle $\hat\vx$ which satisfies $\sum_{{\vx}\in \ucs(\succ_j,\hat\vx)}{p_{\hat j,\vx}} > \sum_{\vx\in \ucs(\succ_j,\hat\vx)}{p_{j,\vx}}$.
	
	Let $t=\sum_{{\vx}\in \ucs(\succ_j,{\hat\vx})}{p_{j,\vx}}$ and $\hat t=\sum_{{\vx} \in \ucs{(\succ_j,\hat\vx)}}{p_{\hat j,\vx}}$. It is easy to see that it must hold that $t< \hat t$. The rest of the proof involves showing that due to the construction of \mps{}, $\hat t\le t$, contradicting our assumption.
	
	Now, let $\vx'$ be the bundle least preferred by $\hat j$ while satisfying $\vx'\in\ucs(\succ_j,\hat\vx)$ and $p_{\hat j,\vx'}>0$. Such a $\vx'$ must exist. Otherwise, $p_{\hat j,\vx}=0$ for every $\vx \in \ucs(\succ_j,\hat\vx)$ which implies that $\hat t=0 \le t$, a contradiction.
	
	Let $t'$ be the time at which $\vx'$ becomes unavailable.
	Due to ${\ucs(\succ_j,\vx') \subseteq \ucs(\succ_j,\hat\vx)}$, we can deduce that $t' \leq \sum_{\vx \in \ucs(\succ_j,\vx')}p_{j,\vx}\leq\sum_{\vx \in \ucs(\succ_j,\hat\vx)}p_{j,\vx} = t$ by (\ref{cm2}) .
	Also, we have $t'=\sum_{\vx \in \ucs(\succ_{\hat j},\vx')}p_{\hat j,\vx}$ by (\ref{cm}).
	By the selection of $\vx'$, we have the following set relation: $\ucs(\succ_j,\hat\vx)\setminus~\{\vx:p_{\hat j,\vx}=~0\}\subseteq~{\ucs(\succ_{\hat j},\vx')}$.	
	Therefore we can deduce that
	\begin{displaymath}
	\hat t =\sum_{\vx \in \ucs(\succ_j,\hat\vx)\setminus \{\vx:p_{\hat j,\vx}=0\}}p_{\hat j,\vx} \leq \sum_{\vx \in \ucs(\succ_{\hat j},\vx')}p_{\hat j,\vx}=t'.
	\end{displaymath}
	This implies $\hat t \leq t' \leq t$, a contradiction.\qed
\end{proof}

\subsection{\textbf{Full Proof of ($\Rightarrow$) part in \Cref{thm:eps}}}\label{sec:thm:eps}
\begin{proof}


    $(\Rightarrow)$ Let $R$ be an arbitrary preference profile, and let $P$ be an arbitrary assignment satisfying \ngc{} w.r.t. $R$. For convenience, we define some quantities to represent the state during the execution of a member of the family of eating algorithms at each round $s$. For ease of exposition, we use $s=0$ to represent the initial state before the start of execution. Let $M^0 = M$, $\md^0 = \md$. Let $B^s = \{ o \in M^{s-1} : \text{there are no } \vx,\allowbreak \hat \vx \in \md^{s-1} \text{ s.t. } o \in \hat \vx \text{ and } (\vx, \hat \vx) \in \gcycle(P) \}$, $M^s = M^{s-1} - B^s$ and $\md^s = \{ \vx \in \md : \text{for every } o \in \vx, o \in M^s \}$ be the available bundles in $M^s$. We note that for any $M^{s-1}\neq \emptyset$, $B^s \neq \emptyset$. Otherwise, for any $o \in M^{s-1}$, there exists $\vx, \hat \vx \in \md^{s-1}$ s.t. $o \in \hat \vx$ and $(\vx, \hat \vx) \in \gcycle(P)$ which implies $\gcycle(P)$ admits a generalized cycle, a contradiction. Hence, $M^s = \emptyset$ for some $s\le np$. Let $S = \min\{s : M^s = \emptyset\}$.
    Let $\omega_j(t)$ be the eating rate of agent $j$ at time $t$, $N(\vx,\md^s)$ be the set of agents who prefer $\vx$ best in the available bundles $\md^s$ for any $\vx \in \md^s$. The following eating speed functions $\omega_j$ define the algorithm:
    \[
    \forall s \leq S, \frac{s-1}{S} \leq t \leq \frac{s}{S}, \omega_j(t)\overset{\text{def}}{=}
    \left\{
    \begin{aligned}
     S \times p_{j,\vx}, \qquad & \exists o \in \vx, o \in B^s \text{ and } j \in N(\vx, \md^{s-1}) \\
     0,  \qquad & \text{otherwise.}
    \end{aligned}
    \right.
    \]
    From the design of algorithm, we know that items in $B^s$ decide which bundles in $\md^{s-1}$ are consumed in round $s$, and these items are not consumed after the round $s$.
    Note that $j \in N(\vx, \md^{s-1})$ implies $\vx \in \md^{s-1}$. We claim that the algorithm specified by the eating speed functions $\omega=(\omega_j)_{j\le n}$ above, outputs $P$ for the \mtap{} with preference profile $R$. Let $Q$ be the output of the eating algorithm. We prove $P = Q$ with induction.

    \noindent{\bf Base case.} For any $o \in B^1$, $\vx \in \md$ s.t. $o \in \vx$, we prove $\forall j \in N, p_{j,\vx} = q_{j,\vx}$. For any $o \in B^1$, suppose
    \[
    \sum_{o\in\vx,\vx\in\md^0}\sum_{j\in N(\vx,\md^0)}p_{j,\vx} < \sum_{o\in\vx,\vx\in\md^0}\sum_{j\in N}p_{j,\vx} = 1.
    \]
    Then, there exists $\vx, \hat \vx \in \md^0$ and $\hat j \in N(\vx, \md^0)$ such that $\vx \succ_{\hat j} \hat \vx$, $o \in \hat \vx$ and $p_{\hat j, \hat \vx} > 0$, which implies  that $o$ is not in $B^1$, a contradiction. Hence, we have
    \[
    \sum_{o\in\vx,\vx\in\md^0}\sum_{j\in N(\vx,\md^0)}p_{j,\vx} = \sum_{o\in\vx,\vx\in\md^0}\sum_{j\in N}p_{j,\vx} = 1.
    \]
    That implies for any $o\in B^1$, $\vx \in \md^0$ s.t. $o \in \vx$, if $j$ is not in $N(\vx, \md^0)$, then $p_{j,\vx} = 0$.

    Then for any $o \in B^1$, $\vx \in \md^0$ s.t. $o \in \vx$ and $j \in N(\vx, \md^0)$, we prove that agent $j$ consumes exactly $p_{j,\vx}$ units of bundle $\vx$. We know that $\vx$ is available in $0 \leq t \leq \frac{1}{S}$ and agent $j$ consumes $\vx$ in $0 \leq t \leq \frac{1}{S}$. Therefore, $q_{j,\vx} \geq \frac{1}{S} \times S \times p_{j,\vx} = p_{j, \vx}$. Hence,
    \[
    \sum_{o\in\vx,\vx\in\md^0}\sum_{j\in N(\vx,\md^0)}q_{j,\vx} \geq
    \sum_{o\in\vx,\vx\in\md^0}\sum_{j\in N(\vx,\md^0)}p_{j,\vx} = 1.
    \]
    Since,
    \[
    \sum_{o\in\vx,\vx\in\md^0}\sum_{j\in N(\vx,\md^0)}q_{j,\vx} \leq \sum_{o\in\vx,\vx\in\md}\sum_{j\in N}q_{j,\vx} = 1,
    \]
    we have
    \[
    \sum_{o\in\vx,\vx\in\md^0}\sum_{j\in N(\vx,\md^0)}q_{j,\vx} =
    \sum_{o\in\vx,\vx\in\md^0}\sum_{j\in N(\vx,\md^0)}p_{j,\vx} = 1.
    \]
    That implies for any $o \in B^1$, $\vx \in \md^0$ s.t. $o \in \vx$,
    \begin{enumerate*}[label=(\roman*)]
    \item if $j \in N(\vx, \md^0)$, $p_{j,\vx} = q_{j, \vx}$;
    \item if $j$ is not in $N(\vx, \md^0)$, $q_{j,\vx} = 0$.
    \end{enumerate*}
    Therefore, for any $o \in B^1$, $\vx \in \md$ s.t. $o \in \vx$ and any $j \in N$, $p_{j,\vx} = q_{j,\vx}$.

    \noindent{\bf Inductive step.} Assume for any $o \in \bigcup_{k = 1}^s B^k$, $\vx \in \md$ s.t. $o \in \vx$ and any $j \in N$, \begin{equation}\label{equ:assumegc}
    p_{j,\vx} = q_{j,\vx}.
    \end{equation}
    We prove for any $o \in B^{s+1}$, $\vx \in \md$ s.t. $o \in \vx$ and any $j \in N$, $p_{j,\vx} = q_{j,\vx}$. If $\vx$ is not in $\md^s$, there is $\hat o \in \bigcup_{k = 1}^s B^k$ such that $\hat o \in \vx$ and we have $p_{j,\vx} = q_{j,\vx}$ by assumption (\ref{equ:assumegc}). For any $o \in B^{s+1}$, suppose
    \[
    \sum_{o\in\vx,\vx\in\md^s}\sum_{j\in N(\vx,\md^s)}p_{j,\vx} + \sum_{o\in\vx,\vx \text{ not in }\md^s}\sum_{j \in N}p_{j,\vx} < \sum_{o\in\vx,\vx\in\md}\sum_{j\in N}p_{j,\vx} = 1.
    \]
    Then, there exists $\vx, \hat \vx \in \md^s$ and $\hat j \in N(\vx, \md^s)$ such that $\vx \succ_{\hat j} \hat \vx$, $o \in \hat \vx$ and $p_{\hat j, \hat \vx} > 0$, which implies that $o$ is not in $B^{s+1}$, a contradiction. Hence, we have
    \[
    \sum_{o\in\vx,\vx\in\md^s}\sum_{j\in N(\vx,\md^s)}p_{j,\vx} + \sum_{o\in\vx,\vx\text{ not in }\md^s}\sum_{j \in N}p_{j,\vx} = \sum_{o\in\vx,\vx\in\md}\sum_{j\in N}p_{j,\vx} = 1.
    \]
    That implies for any $o\in B^{s+1}$, $\vx \in \md^s$ s.t. $o \in \vx$, if $j$ is not in $N(\vx, \md^s)$, then $p_{j,\vx} = 0$.

    For any $o \in B^{s+1}$, $\vx \in \md^s$ s.t. $o \in \vx$ and $j \in N(\vx, \md^s)$, we prove that agent $j$ consumes exactly $p_{j,\vx}$ units of bundle $\vx$. By the assumption (\ref{equ:assumegc}), for any $\hat o \in M^s$, $\hat o$ remains $\sum_{\hat o \in \hat \vx, \hat \vx \in \md^s}\sum_{j\in N(\hat \vx, \md^s)}p_{j,\vx}$ at least at time $t = \frac{s}{S}$. Hence, $\vx$ is available in $\frac{s}{S} \leq t \leq \frac{s+1}{S}$ and agent $j$ consumes $\vx$ in $\frac{s}{S} \leq t \leq \frac{s+1}{S}$. Therefore, $q_{j,\vx} \geq \frac{1}{S} \times S \times p_{j,\vx} = p_{j, \vx}$. Hence,
    \begin{equation*}
        \begin{split}
            \sum_{o\in\vx,\vx\in\md^s}\sum_{j\in N(\vx,\md^s)}q_{j,\vx}  + \sum_{o\in\vx,\vx\text{ not in }\md^s}\sum_{j \in N}q_{j,\vx} \geq \\
    \sum_{o\in\vx,\vx\in\md^s}\sum_{j\in N(\vx,\md^s)}p_{j,\vx} + \sum_{o\in\vx,\vx\text{ not in }\md^s}\sum_{j \in N}p_{j,\vx} = 1.
        \end{split}
    \end{equation*}
    Since,
    \[
    \sum_{o\in\vx,\vx\in\md^s}\sum_{j\in N(\vx,\md^s)}q_{j,\vx}  + \sum_{o\in\vx,\vx\text{ not in }\md^s}\sum_{j \in N}q_{j,\vx} \leq \sum_{o\in\vx,\vx\in\md}\sum_{j\in N}q_{j,\vx} = 1,
    \]
    we have
    \[
    \sum_{o\in\vx,\vx\in\md^s}\sum_{j\in N(\vx,\md^s)}q_{j,\vx}  + \sum_{o\in\vx,\vx\text{ not in }\md^s}\sum_{j \in N}q_{j,\vx} = 1,
    \]
    That implies for any $o \in B^{s+1}$, $\vx \in \md^s$ s.t. $o \in \vx$,
    \begin{enumerate*}[label=(\roman*)]
    \item if $j \in N(\vx, \md^s)$, $p_{j,\vx} = q_{j, \vx}$;
    \item if $j$ is not in $N(\vx, \md^s)$, $q_{j,\vx} = 0$.
    \end{enumerate*}
    Therefore, for any $o \in B^{s+1}$, $\vx \in \md^s$ s.t. $o \in \vx$ and any $j \in N$, $p_{j,\vx} = q_{j,\vx}$.

    By induction, we have for any $o \in \bigcup_{k = 1}^S B^k$, $\vx \in \md$ s.t. $o \in \vx$ and any $j \in N$, $p_{j,\vx} = q_{j,\vx}$. That is to say, for any $\vx \in \md$ and $j \in N$, $p_{j,\vx} = q_{j,\vx}$.\qed
\end{proof}

\subsection{\textbf{Proof of \Cref{cl:mpstops} in \Cref{thm:mpssp}}}\label{sec:cl:mpstops}
\begin{proof}
	We know that each agent does the following things repeatedly in PS: consuming her most preferred and unexhausted item till it is exhausted.
	To prove the claim, we show that agents in \mps{} executes the same in each type as they does in PS, i.e. for any $i\le p,j\le n$, agent $j$ still consumes her most preferred and unexhausted item in type $i$ while consuming her most preferred and available bundle.
	In the following discussion, we use $\bundlep{\vx}{o}{\hat o}$ to denote the bundle replacing $o$ with $\hat o$ in $\vx$.
	
	W.l.o.g. we discuss the behavior of agent $j$ in \mps{} for type $i$.
	At the beginning of \mps{}, the bundle consumed by agent $j$, denoted by $\vx_1$, is her most preferred bundle regarding $\succ_j$.
	By \Cref{dfn:lex}, the most preferred bundle contains agent $j$'s most preferred item in type $i$ regarding $\prefintype{j}{i}$.
	That means when agent $j$ consumes $\vx_1$, she also consumes her most preferred item $o_1=\type{i}{\vx_1}$.
	When the consumption of $\vx_1$ pauses, agent $j$ turns to the bundle consumed after $\vx_1$, denoted by $\vx_2$.
	We note that $\vx_2$ does not need to be the second preferred regarding $\succ_j$.
	There are two kinds of cases before consuming $\vx_2$:
	\begin{enumerate*}[label=(\roman*)]
		\item $\type{i_0}{\vx_1}$ is exhausted, $i_0 \neq i$,
		\item $o_1$ is exhausted.
	\end{enumerate*}
	
	We claim that $\type{i}{\vx_2}$ is the most preferred and unexhausted item in type $i$ just after $\vx_1$ is unavailable.
	For Case~(\romannumeral1), the agent $j$'s most preferred item in type $i$ is still $o_1$.
	We claim that $o_1 \in \vx_2$. Otherwise, suppose that $o_2=\type{i}{\vx_2} \neq o_1$.
	Then $o_1 \prefintype{j}{i} o_2$.
	Because $\type{i_0}{(o_1,x_2\backslash o_2)} = \type{i_0}{\vx_2}$ for $i_0 \impord{j} i$ and $o_1=\type{i}{(o_1,x\backslash o_2)} \prefintype{j}{i} \type{i}{\vx_2}=o_2$, it follows that $\bundlep{\vx_2}{o_2}{o_1} \succ_j \vx_2$, which is a contradiction to the fact that $\vx_2$ is the next preferred bundle for agent $j$.
	For Case~(\romannumeral2), let $M_1$ be the set of unexhausted items just after $\vx_1$ is unavailable in \mps{}, and $o_2\in M_1\bigcap D_i$ be agent $j$'s most preferred and unexhausted item in $D_i$ regarding $\prefintype{j}{i}$.
	We also note that $o_2$ does not need to be the second preferred item regarding $\prefintype{j}{i}$.
	By \Cref{alg:mps}, we know that $o_2=\type{i}{\vx_2}$, and we can obtain that ${\vx_2}\succ_j\bundlep{{\vx_2}}{o_2}{o'_2}$ if $o'_2 \in M_1\bigcap D_i$ and $o'_2\neq o_2$.
	
	The claim in the previous paragraph can be applied to the general case when $\vx_k$ is unavailable and turns to $\vx_{k+1}$, and we have that $\type{i}{\vx_{k+1}}$ is always the most preferred and unexhausted item in type $i$ just after $\vx_k$ is unavailable.
	We note again that $\vx_k$ does not need to be the $k$th preferred regarding $\succ_j$, and  $\vx_{k+1}$ only refers to the bundle consumed after $\vx_k$.
	From the result of proof, we can observe that agent $j$ consumes the bundle containing the most preferred and unexhausted item in type $i$.
	That is exactly what agent $j$ does in PS of type $i$.
	This result can be extended to any agent and any type.
	Thus each single type \ram{} assignment of $P$ is the same as the one produced by PS of that type.
\end{proof}
\subsection{\textbf{Proof of \Cref{cl:eqwtype} in \Cref{thm:mpssp}}}\label{sec:cl:eqwtype}
\begin{proof}
	
	When $\hat i\le i'=1$ i.e. $\srd{P}{1}_j=\srd{Q}{1}_j$, the claim means for $S_{1}=\{(\itemwt{1}*)|\itemwt{1}\in\ucs(\prefintype{j}{1},\hitemwt{1})\}$,
	\begin{equation*}
	    \sum_{\vx\in S_{1}}{p_{j,{\vx}}}=\sum_{o\in\ucs(\prefintype{j}{1},\hitemwt{1})}{\srd{p}{1}_{j,o}}=\sum_{o\in\ucs(\prefintype{j}{1},\hitemwt{1})}{\srd{q}{1}_{j,o}}=\sum_{\vx\in S_{1}}{q_{j,{\vx}}},
	\end{equation*}
	which is trivially true.

	Given $\srd{P}{\hat i}=\srd{Q}{\hat i}$ for $\hat i\le i'$ and the claim is true for $\hat i\le i'-1$, which also means
	\begin{equation}\label{eq3:cl:eqwtype}
	    \sum_{\vx\in\{(\hitemwt{1},\dots,\hitemwt{\hat i -1},\itemwt{\hat i}*)|\itemwt{\hat i}\prefintype{j}{\hat i}\hitemwt{\hat i}\}}p_{j,\vx}=
	    \sum_{\vx\in\{(\hitemwt{1},\dots,\hitemwt{\hat i -1},\itemwt{\hat i}*)|\itemwt{\hat i}\prefintype{j}{\hat i}\hitemwt{\hat i}\}}q_{j,\vx},
	\end{equation}
	we prove that it is also true for $i'$.
	Let $\hat\vx$ denote the least preferred bundle containing $\hitemwt{1},\dots,\hitemwt{i'}$.
	Since $\srd{P}{\hat i}_j=\srd{Q}{\hat i}_j$ for $\hat i\le i'$, we have that $\sum_{\vx\in\ucs(\succ_j,\hat\vx)}p_{j,\vx}\ge\sum_{\vx\in\ucs(\succ_j,\hat\vx)}q_{j,\vx}$.
	$min_{\hitemwt{i}\in\vx,i\le i'}\sum_{o\in\ucs(\prefintype{j}{i},\hitemwt{i})}\srd{p}{i}_o$
	That is because agent $j$ consumes bundles in the set $\ucs(\succ_j,\hat\vx)$ until they are unavailable in $P$, and therefore she is unable to get more on this set given the same shares of items.
	Because $Q\sd{j}P$, we have $\sum_{\vx\in\ucs(\succ_j,\hat\vx)}p_{j,\vx}=\sum_{\vx\in\ucs(\succ_j,\hat\vx)}q_{j,\vx}$.
	We also have that $\sum_{\vx\in\ucs(\succ_j,\hat\vx)}p_{j,\vx}=\sum_{\vx\in\ucs(\succ_j,\hat\vx)}q_{j,\vx}$, which is equal to
	
	\begin{equation}\label{eq2:cl:eqwtype}
	\begin{split}
	    &\sum_{\vx\in\{(\itemwt{1}*)|\itemwt{1}\prefintype{j}{1}\hitemwt{1}\}}p_{j,\vx}+\sum_{\vx\in\{(\hitemwt{1},\itemwt{2}*)|\itemwt{2}\prefintype{j}{2}\hitemwt{2}\}}p_{j,\vx}+\dots+\sum_{\vx\in\{(\hitemwt{1},\hitemwt{2},\itemwt{i-1}*)|\itemwt{i-1}\prefintype{j}{i-1}\hitemwt{i-1}\}}q_{j,\vx}+\sum_{\vx\in S_{i'}}p_{j,\vx}\\
	    =&\sum_{\vx\in\{(\itemwt{1}*)|\itemwt{1}\prefintype{j}{1}\hitemwt{1}\}}q_{j,\vx}+\sum_{\vx\in\{(\hitemwt{1},\itemwt{2}*)|\itemwt{2}\prefintype{j}{2}\hitemwt{2}\}}q_{j,\vx}+\dots+\sum_{\vx\in\{(\hitemwt{1},\hitemwt{2},\itemwt{i-1}*)|\itemwt{i-1}\prefintype{j}{i-1}\hitemwt{i-1}\}}q_{j,\vx}+\sum_{\vx\in S_{i'}}q_{j,\vx}.
	\end{split}
	\end{equation}

    Here $S_{i'}=\{(\hitemwt{1},\dots,\hitemwt{i'-1}\itemwt{i'}*)|\itemwt{i'}\in\ucs(\prefintype{j}{i'},\hitemwt{i'})\}$.
	Therefore by subtracting (\ref{eq3:cl:eqwtype}) for all $\hat i<i'$ from (\ref{eq2:cl:eqwtype}) we have that $\sum_{\vx\in S_{i'}}p_{j,\vx}=\sum_{\vx\in S_{i'}}q_{j,\vx}$.\qed
\end{proof}

%
%

\bibliography{Manuscript}

\begin{thebibliography}{47}
\providecommand{\natexlab}[1]{#1}
\providecommand{\url}[1]{{#1}}
\providecommand{\urlprefix}{URL }
\expandafter\ifx\csname urlstyle\endcsname\relax
  \providecommand{\doi}[1]{DOI~\discretionary{}{}{}#1}\else
  \providecommand{\doi}{DOI~\discretionary{}{}{}\begingroup
  \urlstyle{rm}\Url}\fi
\providecommand{\eprint}[2][]{\url{#2}}

\bibitem[{Abdulkadiro\u{g}lu and S\"{o}nmez(1998)}]{Abdulkadiroglu98:Random}
Abdulkadiro\u{g}lu A, S\"{o}nmez T (1998) {Random serial dictatorship and the
  core from random endowments in house allocation problems}. Econometrica
  66(3):689--702

\bibitem[{Abdulkadiro\u{g}lu and S\"{o}nmez(1999)}]{Abdulkadiroglu99:House}
Abdulkadiro\u{g}lu A, S\"{o}nmez T (1999) {House allocation with existing
  tenants}. Journal of Economic Theory 88:233--260

\bibitem[{Athanassoglou and Sethuraman(2011)}]{Athanassoglou11:House}
Athanassoglou S, Sethuraman J (2011) House allocation with fractional
  endowments. International Journal of Game Theory 40(3):481--513

\bibitem[{Aziz and Stursberg(2014)}]{Aziz14:Generalization}
Aziz H, Stursberg P (2014) A generalization of probabilistic serial to
  randomized social choice. In: Proceedings of the Twenty-Eighth AAAI
  Conference on Artificial Intelligence, AAAI Press, AAAI’14, pp 559–--565

\bibitem[{Birkhoff(1946)}]{Birkhoff1946:Three}
Birkhoff G (1946) {Three observations on linear algebra}. Universidad Nacional
  de Tucuman, Revista, Serie A 5:147--151

\bibitem[{Bogomolnaia(2015)}]{Bogomolnaia15:Random}
Bogomolnaia A (2015) Random assignment: redefining the serial rule. Journal of
  Economic Theory 158:308--318

\bibitem[{Bogomolnaia and Heo(2012)}]{Bogomolnaia12:Probabilistic}
Bogomolnaia A, Heo EJ (2012) Probabilistic assignment of objects:
  characterizing the serial rule. Journal of Economic Theory 147(5):2072--2082

\bibitem[{Bogomolnaia and Moulin(2001)}]{Bogomolnaia01:New}
Bogomolnaia A, Moulin H (2001) {A new solution to the random assignment
  problem}. Journal of Economic Theory 100(2):295--328

\bibitem[{Bogomolnaia et~al.(2005)Bogomolnaia, Moulin, and
  Stong}]{bogomolnaia2005collective}
Bogomolnaia A, Moulin H, Stong R (2005) Collective choice under dichotomous
  preferences. Journal of Economic Theory 122(2):165--184

\bibitem[{Booth et~al.(2010)Booth, Chevaleyre, Lang, Mengin, and
  Sombattheera}]{Booth10:Learning}
Booth R, Chevaleyre Y, Lang J, Mengin J, Sombattheera C (2010) Learning
  conditionally lexicographic preference relations. In: Proceedings of the 2010
  Conference on ECAI 2010: 19th European Conference on Artificial Intelligence,
  IOS Press, NLD, p 269–274

\bibitem[{Boutilier et~al.(2004)Boutilier, Brafman, Domshlak, Hoos, and
  Poole}]{Boutilier04:CP}
Boutilier C, Brafman R, Domshlak C, Hoos H, Poole D (2004) {CP}-nets: {A} tool
  for representing and reasoning with conditional ceteris paribus statements.
  Journal of artificial intelligence research 21:135--191

\bibitem[{Brams and Taylor(1996)}]{brams1996fair}
Brams SJ, Taylor AD (1996) Fair fivision: from cake-cutting to dispute
  resolution. Cambridge University Press

\bibitem[{Chevaleyre et~al.(2006)Chevaleyre, Dunne, Endriss, Lang, Lemaitre,
  Maudet, Padget, Phelps, Rodr\'{i}guez-Aguilar, and
  Sousa}]{Chevaleyre06:Issues}
Chevaleyre Y, Dunne PE, Endriss U, Lang J, Lemaitre M, Maudet N, Padget J,
  Phelps S, Rodr\'{i}guez-Aguilar JA, Sousa P (2006) Issues in multiagent
  resource allocation. Informatica 30:3--31

\bibitem[{Edmonds and Pruhs(2006)}]{edmonds2006cake}
Edmonds J, Pruhs K (2006) Cake cutting really is not a piece of cake. In:
  Proceedings of the seventeenth annual ACM-SIAM symposium on Discrete
  algorithm, Society for Industrial and Applied Mathematics, pp 271--278

\bibitem[{Even and Paz(1984)}]{even1984note}
Even S, Paz A (1984) A note on cake cutting. Discrete Applied Mathematics
  7(3):285--296

\bibitem[{Fujita et~al.(2015)Fujita, Lesca, Sonoda, Todo, and
  Yokoo}]{Fujita2015:A-Complexity}
Fujita E, Lesca J, Sonoda A, Todo T, Yokoo M (2015) A complexity approach for
  core-selecting exchange with multiple indivisible goods under lexicographic
  preferences. In: Proceedings of the Twenty-Ninth AAAI Conference on
  Artificial Intelligence, AAAI Press, AAAI’15, p 907–913

\bibitem[{Gale and Shapley(1962)}]{Gale62:College}
Gale D, Shapley LS (1962) College admissions and the stability of marriage. The
  American Mathematical Monthly 69(1):9--15

\bibitem[{Ghodsi et~al.(2011)Ghodsi, Zaharia, Hindman, Konwinski, Shenker, and
  Stoica}]{Ghodsi11:Dominant}
Ghodsi A, Zaharia M, Hindman B, Konwinski A, Shenker S, Stoica I (2011)
  {Dominant resource fairness: fair allocation of multiple resource types}. In:
  {Proceedings of the 8th USENIX Conference on Networked Systems Design and
  Implementation}, {Boston, MA, USA}, pp 323--336

\bibitem[{Ghodsi et~al.(2012)Ghodsi, Sekar, Zaharia, and
  Stoica}]{Ghodsi12:Multi}
Ghodsi A, Sekar V, Zaharia M, Stoica I (2012) Multi-resource fair queueing for
  packet processing. In: Proceedings of the ACM SIGCOMM 2012 Conference on
  Applications, Technologies, Architectures, and Protocols for Computer
  Communication, Association for Computing Machinery, New York, NY, USA,
  SIGCOMM ’12, p 1–12

\bibitem[{Grandl et~al.(2015)Grandl, Ananthanarayanan, Kandula, Rao, and
  Akella}]{Grandl15:Multi}
Grandl R, Ananthanarayanan G, Kandula S, Rao S, Akella A (2015) Multi-resource
  packing for cluster schedulers. ACM SIGCOMM Computer Communication Review
  44(4):455--466

\bibitem[{Hashimoto and Hirata(2011)}]{hashimoto2011characterizations}
Hashimoto T, Hirata D (2011) Characterizations of the probabilistic serial
  mechanism. Available at SSRN \doi{http://dx.doi.org/10.2139/ssrn.1601178}

\bibitem[{Hashimoto et~al.(2014)Hashimoto, Hirata, Kesten, Kurino, and
  {\"U}nver}]{Hashimoto14:Two}
Hashimoto T, Hirata D, Kesten O, Kurino M, {\"U}nver MU (2014) Two axiomatic
  approaches to the probabilistic serial mechanism. Theoretical Economics
  9(1):253--277

\bibitem[{Hatfield(2009)}]{Hatfield09:Strategy-proof}
Hatfield JW (2009) Strategy-proof, efficient, and nonbossy quota allocations.
  Social Choice and Welfare 33(3):505--515

\bibitem[{Heo(2014)}]{Heo14:Probabilistic}
Heo EJ (2014) Probabilistic assignment problem with multi-unit demands: A
  generalization of the serial rule and its characterization. Journal of
  Mathematical Economics 54:40--47

\bibitem[{Heo and Y{\i}lmaz(2015)}]{Heo15:Characterization}
Heo EJ, Y{\i}lmaz {\"O} (2015) A characterization of the extended serial
  correspondence. Journal of Mathematical Economics 59:102--110

\bibitem[{Hylland and Zeckhauser(1979)}]{Hylland79:Efficient}
Hylland A, Zeckhauser R (1979) {The Efficient Allocation of Individuals to
  Positions}. Journal of Political Economy 87(2):293--314

\bibitem[{Katta and Sethuraman(2006)}]{Katta06:Solution}
Katta AK, Sethuraman J (2006) A solution to the random assignment problem on
  the full preference domain. Journal of Economic theory 131(1):231--250

\bibitem[{Kojima and Manea(2010)}]{Kojima09:Axioms}
Kojima F, Manea M (2010) Axioms for deferred acceptance. Econometrica
  78(2):633--653

\bibitem[{Lang et~al.(2012)Lang, Mengin, and Xia}]{Lang12:Aggregating}
Lang J, Mengin J, Xia L (2012) Aggregating conditionally lexicographic
  preferences on multi-issue domains. In: Proceedings of the 18th International
  Conference on Principles and Practice of Constraint Programming (CP), pp
  973--987

\bibitem[{Mackin and Xia(2016)}]{Mackin2016:Allocating}
Mackin E, Xia L (2016) {Allocating indivisible items in categorized domains}.
  In: Proceedings of the Twenty-Fifth International Joint Conference on
  Artificial Intelligence (IJCAI-16), pp 359--365

\bibitem[{Moulin(1995)}]{Moulin95:Cooperative}
Moulin H (1995) {Cooperative microeconomics: a game-theoretic introduction}.
  Prentice Hall

\bibitem[{Moulin(2018)}]{moulin2018fair}
Moulin H (2018) Fair division in the age of internet. Annual Review of
  Economics

\bibitem[{von Neumann(1953)}]{Neumann1953:A-certain}
von Neumann J (1953) A certain zero-sum two-person game equivalent to an
  optimal assignment problem. Annals of Mathematics Studies 2(28):5--12

\bibitem[{Procaccia(2015)}]{procaccia2015cake}
Procaccia AD (2015) Cake cutting algorithms. In: Handbook of Computational
  Social Choice, chapter 13, Cambridge University Press

\bibitem[{Robertson and Webb(1998)}]{robertson1998cake}
Robertson J, Webb W (1998) Cake-cutting algorithms: be fair if you can. CRC
  Press

\bibitem[{Roth and Postlewaite(1977)}]{Roth77:Weak}
Roth AE, Postlewaite A (1977) {Weak versus strong domination in a market with
  indivisible goods}. Journal of Mathematical Economics 4:131--137

\bibitem[{Saban and Sethuraman(2014)}]{Saban14:Note}
Saban D, Sethuraman J (2014) A note on object allocation under lexicographic
  preferences. Journal of Mathematical Economics 50:283--289

\bibitem[{Segal-Halevi(2017)}]{Segal17:Fair}
Segal-Halevi E (2017) Fair division of land. PhD thesis, Bar Ilan University,
  Computer Science Department

\bibitem[{Shapley and Scarf(1974)}]{Shapley74:Cores}
Shapley L, Scarf H (1974) On cores and indivisibility. Journal of Mathematical
  Economics 1(1):23--37

\bibitem[{Sikdar et~al.(2017)Sikdar, Adalundefined, and
  Xia}]{Sikdar2017:Mechanism}
Sikdar S, Adalundefined S, Xia L (2017) Mechanism design for multi-type housing
  markets. In: Proceedings of the Thirty-First AAAI Conference on Artificial
  Intelligence, AAAI Press, AAAI’17, p 684–690

\bibitem[{Sikdar et~al.(2019)Sikdar, Adal{\i}, and Xia}]{Sikdar18:Top}
Sikdar S, Adal{\i} S, Xia L (2019) Mechanism design for multi-type housing
  markets with acceptable bundles. In: Proceedings of the Thirty-Third AAAI
  Conference on Artificial Intelligence, AAAI Press, AAAI’19, pp 2165--2172

\bibitem[{Steinhaus(1948)}]{Steinhaus48:Problem}
Steinhaus H (1948) The problem of fair division. Econometrica 16:101--104

\bibitem[{Wang et~al.(2020)Wang, Sikdar, Guo, Xia, Cao, and
  Wang}]{Wang19:Multi}
Wang H, Sikdar S, Guo X, Xia L, Cao Y, Wang H (2020) Multi-type resource
  allocation with partial preferences. In: Proceedings of the Thirty-Fourth
  AAAI Conference on Artificial Intelligence, AAAI Press, AAAI’20

\bibitem[{Xia and Conitzer(2010)}]{Xia10:Strategy}
Xia L, Conitzer V (2010) Strategy-proof voting rules over multi-issue domains
  with restricted preferences. In: Proceedings of the 6th International
  Conference on Internet and Network Economics, Springer-Verlag, Berlin,
  Heidelberg, WINE’10, p 402–414

\bibitem[{Yilmaz(2009)}]{Yilmaz2009:Random}
Yilmaz {\"O} (2009) Random assignment under weak preferences. Games and
  Economic Behavior 66(1):546--558

\bibitem[{Y{\i}lmaz(2010)}]{Yilmaz10:Probabilistic}
Y{\i}lmaz {\"O} (2010) The probabilistic serial mechanism with private
  endowments. Games and Economic Behavior 69(2):475--491

\bibitem[{Zhou(1990)}]{zhou1990conjecture}
Zhou L (1990) On a conjecture by gale about one-sided matching problems.
  Journal of Economic Theory 52(1):123--135

\end{thebibliography}
\bibliographystyle{spbasic}      


\end{document}